\documentclass[acmsmall,screen,review,nonacm]{acmart}

\renewcommand\footnotetextcopyrightpermission[1]{}
\setcopyright{none}
\AtBeginDocument{%
  \fancypagestyle{standardpagestyle}{%
    \fancyhf{}%
    \fancyfoot[RO,LE]{\small\thepage}%
  }%
  \fancypagestyle{firstpagestyle}{%
    \fancyhf{}%
    \fancyfoot[RO,LE]{\small\thepage}%
  }%
  \pagestyle{standardpagestyle}%
}

\usepackage{fontawesome}
\usepackage{balance}
\usepackage{tikz}
\usepackage{quantikz}
\usepackage{amsfonts}
\usepackage{pifont}
\usepackage{xspace}
\usepackage{graphicx}
\usepackage{wrapfig}
\usepackage[capitalise]{cleveref}
\usepackage{thmtools}
\theoremstyle{acmplain}
\declaretheorem[name=Theorem,numberwithin=section]{theorem}

\declaretheorem[name=Lemma,sibling=theorem]{lemma}

\theoremstyle{acmdefinition}
\declaretheorem[name=Example,sibling=theorem]{example}

\theoremstyle{acmplain}

\usepackage{subcaption}
\usepackage{tcolorbox}
\usepackage{algorithm}
\usepackage[noend]{algpseudocode}
\usepackage{numprint}
\usepackage{enumitem}
\setlist[description]{leftmargin=\parindent,labelindent=\parindent}

\npthousandsep{,}

\Crefname{figure}{Fig.}{Figs.}
\Crefname{tabular}{Tab.}{Tabs.}
\Crefname{section}{\S}{\S}
\Crefname{lemma}{Lem.}{Lems.}
\Crefname{theorem}{Thm.}{Thms.}
\Crefname{corollary}{Cor.}{Cors.}
\Crefname{algorithm}{Alg.}{Algs.}
\Crefname{example}{Ex.}{Exs.}
\Crefname{definition}{Def.}{Defs.}

\newcommand{\ours}{\textsc{tzap}\xspace}
\newcommand{\ourssymb}{\textsc{tzap-symb}\xspace}
\newcommand{\ourspf}{\textsc{tzap-pf}\xspace}
\newcommand{\pyzx}{PyZX\xspace}
\newcommand{\cliffordt}{Clifford+$T$\xspace}

\newcommand{\circuit}{\ensuremath{\mathcal{C}}\xspace}

\newcommand{\CX}{\ensuremath{\mathit{CX}}\xspace}
\newcommand{\X}{\ensuremath{\mathit{X}}\xspace}
\newcommand{\Had}{\ensuremath{H}\xspace}

\newcommand{\Rz}[1]{\ensuremath{\mathit{Rz}(#1)}}

\newcommand{\sem}[1]{\ensuremath{\llbracket #1 \rrbracket}}
\newcommand{\rasem}[1]{\ensuremath{\llbracket #1 \rrbracket^{\#}}}
\newcommand{\sasem}[1]{\ensuremath{\llbracket #1 \rrbracket^{\hat{\#}}}}

\definecolor{qubitblue}{rgb}{0.0, 0.0, 0.55}
\newcommand{\q}[1]{{\color{qubitblue}{#1}}}
\renewcommand{\read}{\ensuremath{\mathsf{read}}\xspace}
\newcommand{\true}{\ensuremath{\mathsf{true}}\xspace}
\newcommand{\false}{\ensuremath{\mathsf{false}}\xspace}
\newcommand{\negate}{\ensuremath{\mathsf{neg}}\xspace}
\newcommand{\ifflip}[2]{\mathsf{if}~#1~\negate(#2)\xspace}
\definecolor{hoareblue}{rgb}{0.0, 0.0, 0.7}
\definecolor{maroon}{rgb}{0.5, 0.0, 0.0}
\newcommand{\fresh}[1]{\setlength{\fboxsep}{1pt}\colorbox{gray!25}{$#1$}}
\newcommand{\hoare}[1]{{\color{hoareblue}\left[#1\right]}}
\newcommand{\phoare}[2]{{\color{hoareblue}[#1]}_{\color{maroon}#2}}
\providecommand{\ket}{}
\renewcommand{\ket}[1]{|#1\rangle}

\newtcolorbox{takeaway}{
  colback=black!5,
  colframe=black!5,
  boxrule=0pt,
  arc=2pt,
  left=6pt, right=6pt, top=4pt, bottom=4pt,
  before upper={\textbf{Summary.}\hspace{0.5em}\ignorespaces},
}
\begin{document}

\title{Linear-Time $T$-Gate Optimization via Random Abstraction}
\author{Aws Albarghouthi}
\affiliation{%
  \institution{University of Wisconsin-Madison}%
  \city{Madison, WI}
  \country{USA}%
}
\email{aws@cs.wisc.edu}

\begin{abstract}
Quantum computers promise exponential speedups for problems in cryptography, chemistry, and optimization.
Realizing this promise requires fault tolerance: physical qubits are noisy, so logical qubits must be encoded redundantly across many physical ones using quantum error-correcting codes.
In most practical fault-tolerance schemes, $T$ gates---phase rotations by $\pi/4$---cannot be implemented transversally and instead require costly magic-state distillation protocols involving a complex set of operations.
As a result, $T$-gate count can dominate the resource budget of large-scale quantum computations, making $T$-count minimization a central bottleneck on the path to quantum advantage.
Existing $T$-count optimization tools, however, do not scale to the circuits that quantum advantage demands.

We present theoretical and practical results on $T$-gate optimization.
On the theoretical side, we give a linear-time randomized algorithm for phase folding, based on a novel randomized static analysis.
Our static analysis approximates the transition relation of a quantum circuit and is sound except with an arbitrarily small probability of error.
Our key insight is a static analysis that does not track symbolic expressions, but propagates constant-width bitstrings down the circuit.
On the practical side, our implementation, \ours, is multiple orders of magnitude faster than state-of-the-art tools---such as \pyzx, VOQC, and Feynman---closely matches their $T$-count reductions on standard benchmarks, and can optimize circuits with millions of gates within seconds on a laptop computer.
~\\~\\
\noindent
\faGithub\ \ \ \url{https://github.com/qqq-wisc/tzap}
\end{abstract}

\maketitle
\section{Introduction}
Quantum computers promise exponential speedups for a range of practically important problems---from simulating quantum chemistry~\cite{lloyd1996universal} to breaking public-key cryptography via Shor's algorithm~\cite{shor1994algorithms}.
Realizing this power requires \emph{fault tolerance}: physical qubits are inherently noisy, so logical qubits must be encoded redundantly across many physical ones using quantum error-correcting codes, and gates must be applied in a fault-tolerant manner.

In most practical fault-tolerance schemes, we use the universal \cliffordt gate set. 
While Clifford gates are efficiently simulable classically and are typically much cheaper to implement fault-tolerantly, $T$ gates---phase rotations by $\pi/4$---cannot be implemented transversally in many practical fault-tolerance schemes, and instead require costly \emph{magic-state-distillation} protocols~\cite{bravyi2005universal,bravyi2012magic,gidney2024cultivation} involving an elaborate and complex sequence of operations.
Each logical $T$ gate can demand thousands of physical operations.
As a result, $T$-gate count can dominate the resource budget of large-scale quantum computations, and minimizing the number of $T$ gates in a circuit is a central bottleneck on the path to quantum advantage.

The challenge of $T$-count reduction has attracted substantial attention, and a range of techniques have been developed.
One broad family of approaches relies on \emph{phase folding}: when two rotation gates act on the same basis states, their phases can be combined and redundant gates eliminated.
This idea has appeared in different guises across the literature~\cite{amy2014polynomial,nam2018automated,amy2025linear}, and has been shown to be highly effective in practice.
A second family of approaches uses \emph{rewriting}: ZX-calculus represents a quantum circuit as a graphical tensor-network diagram and applies a suite of local rewrite rules to reduce gate count~\cite{coecke2011interacting}.
Tools based on these ideas, such as Feynman~\cite{amy2014polynomial,amy2025linear} and \pyzx~\cite{kissinger2020pyzx}, achieve strong $T$-count reductions on standard benchmarks.

\subsection{Fast, Scalable $T$-Gate Optimization}
We study $T$-count minimization with two goals in mind.
On the \emph{theoretical} side, we seek a \emph{linear-time} algorithm.
Quantum advantage applications may involve circuits with billions of gates~\cite{babbush2026grand}, and any super-linear algorithm will be prohibitively slow at that scale.
On the \emph{practical} side, we want an approach that matches the reduction quality of existing state-of-the-art tools---fast runtime must not come at the cost of optimization power.
To this end, we present a linear-time randomized algorithm for phase folding, whose key ingredient is a lightweight randomized static analysis of the circuit.

\paragraph{The phase-folding optimization}
The heart of phase folding is proving that two rotation gates act on the same qubit state at two distinct points in the circuit.
For intuition, consider an analogy from classical program optimization.
Suppose a compiler sees two array increment statements, \textsf{A[x]++} and \textsf{A[y]++}, and wants to fuse them into a single \textsf{A[x] += 2}.
This is only sound if \textsf{x = y} \emph{for every possible execution of the program}---regardless of how the initial variables are set.
A static analysis can establish this by tracking a symbolic expression for each variable as statements execute: if the analysis can prove that the expressions for \textsf{x} and \textsf{y} are always equal, the fusion is justified.

Phase folding in quantum circuits is the same idea: two rotation gates can be merged---and one eliminated---when they act on the same basis states for all possible inputs.
Characterizing the set of reachable quantum states at every rotation is generally intractable.
Instead, existing approaches soundly approximate the set of states by, for example, only tracking the effects of classical gates~\cite{xu2022synthesizing} and treating others (like Hadamard) as non-deterministic operations~\cite{amy2025linear}. We take great inspiration from this view and use it to formalize a probabilistically sound and efficient abstraction.

\paragraph{Fast phase folding with randomized abstraction}
Our key insight is that the \emph{transition relation} of a quantum circuit can be approximated with a special algebraic structure: every tracked qubit value is a \emph{linear} function over Boolean variables for the initial qubit values and fresh outcomes introduced at Hadamards---it is a parity.
The state-equivalence question therefore reduces to a parity-equivalence question, which admits a simple and efficient randomized solution: Evaluate the two parities on random bitstrings and check if they evaluate to the same answer!

Randomized parity equivalence provides us with a key ingredient for building a very fast static analysis over quantum circuits:
Rather than maintaining symbolic expressions, we assign each qubit a random bitstring at the start of the circuit and propagate these bitstrings forward through each gate---one inexpensive bit-level operation per gate.
These bitstrings constitute an \emph{abstract domain} that probabilistically overapproximates the transition relation of a quantum circuit, achieving the same goal as the symbolic analysis with constant-time, fast, bit-level operations.
Each bit position in the bitstring is an independent random draw, so a width-$k$ bitstring encodes $k$ independent random trials; this is what makes the error probability drop exponentially with $k$.
One can view our approach as an analog of the \emph{random interpretation} framework~\cite{gulwani2003discovering} for quantum circuits, where we use randomization to achieve a very fast static analysis that can implicitly track affine equalities.

Using our randomized abstract domain, we can propagate facts down the circuit and merge rotations that share the same bitstrings.
The probability of a spurious merge between any specific pair of rotations is exponentially small in the bitstring width, but a correct run requires \emph{no} spurious merges across all pairs---a \emph{birthday-paradox} situation. By choosing the bitstring width large enough relative to the circuit size, the probability of an unsound merge is negligible.
Overall, we demonstrate that our optimization algorithm runs in expected time that is linear in the size of the circuit.
In other words, the algorithm performs one linear scan through the circuit, eliminating unnecessary rotations. As we demonstrate experimentally, this allows it to easily scale to very large circuits with millions of gates.

\subsection{Summary of Experimental Results}
We have implemented our approach in a tool called \ours and evaluated it
against five state-of-the-art optimizers spanning the main
families of $T$-count reduction techniques: VOQC~\cite{hietala2021verified},
a verified optimizer that applies the Nam et al.~\cite{nam2018automated}
rewriting and pass-scheduling protocol; \pyzx~\cite{kissinger2020pyzx} and
its Rust reimplementation QuiZX~\cite{kissinger2021quizx}, which optimize via ZX-calculus rewriting
and graph extraction; Feynman~\cite{amy2014polynomial,amy2025linear},
whose affine analysis is the closest in spirit to our
approach; and FastTODD~\cite{vandaele2024lower}, which uses Reed--Muller decoding to optimize $T$-count.

We compare on two benchmark suites. The first is the suite
that comes with Feynman and has become the \emph{de facto} standard for
$T$-count optimization, comprising arithmetic and oracle circuits from the
quantum-algorithms literature, with sizes ranging from tens of gates up to
hundreds of thousands.
The second is the Cobble suite~\cite{yuan2024cobble}, which contains larger and more structurally
varied circuits drawn from Hamiltonian simulation and quantum linear-algebra
kernels, with sizes ranging from a few thousand gates up to hundreds of
millions of gates after \cliffordt decomposition---stressing scalability
rather than just optimization quality.

Across both suites, \ours closely matches the $T$-count reductions of these tools, and typically matches them exactly on the largest circuits.
Notably, \ours runs four orders of magnitude faster than the best competing
tool on large circuits. Beyond a few hundred thousand gates, \ours finishes in milliseconds and the other tools time out after one hour.
We also conduct scaling experiments on parameterized
families of circuits with up to hundreds of millions of gates. \ours scales
linearly with circuit size, optimizing instances with hundreds of millions
of gates in minutes on a laptop, while consistently delivering $T$-count
reductions of $50\%$ to $70\%$.

\subsection{Contributions}
The main contributions of this work are:
\begin{itemize}
  \item We define a novel, randomized static analysis over an abstract domain of bitstrings that overapproximates the transition relation of a quantum circuit, sound except with an arbitrarily small probability of error.
  \item We present a linear-time phase-folding algorithm that uses the randomized static analysis to decide, with high-probability guarantees, whether two rotation gates are mergeable via a single forward pass through the circuit.
  \item We implement our approach and perform a thorough experimental evaluation demonstrating orders-of-magnitude runtime improvements over existing tools while closely matching their output quality.
\end{itemize}

\section{Overview of Our Approach}
\label{sec:overview}
We now present an overview of our approach for optimizing quantum circuits. The approach hinges on a static analysis that captures the transition relation of the program: how the state at each point in the program relates to the initial state.
\subsection{Let's Start with a Simple Language}
To illustrate and motivate our approach within the broader notion of static analysis, we will begin by examining the static analysis question for a very simple programming language.

\paragraph{A simple PL}
Suppose we have a language with the following features: A program is over a finite set of Boolean variables $x_1, x_2, \ldots, x_n$.
A program is a straight-line sequence of instructions (no loops), where each instruction is one of the following:
\begin{itemize}
\item \textbf{Negate a bit}: $\negate(x_i)$
\item \textbf{Read bit}: $\read(x_i)$
\item \textbf{Conditional}: $\ifflip{x_i}{x_j}$
\end{itemize} 

So, for example, we can write the following program:
\begin{align*}
  &\negate(x_1) \\
  &\read(x_2) \\
  &\ifflip{x_2}{x_1}
\end{align*}
Initially, each $x_i$ is either $\true$ or $\false$ (we treat the initial state as non-deterministic input).
Here, $\read(x_2)$ overwrites the value in $x_2$ with an input from the user.
So, this program negates $x_1$ (flips the bit), reads a fresh value into $x_2$, and then negates $x_1$ again if $x_2$ is true.

\paragraph{Statically analyzing programs}
Like any imperative programming language, we can reason about programs in this language statically, by propagating facts forwards through instructions.
Consider a program that has a single variable, $x_1$.
We can represent the set of all possible initial states of this program as follows:
$$\hoare{x_1 \mapsto \nu_1}$$
where $\nu_1$ is a symbolic variable denoting $x_1$'s initial value.
Since $\nu_1$ is unconstrained, $x_1$ can take any value.
In other words, we treat $\hoare{x_1 \mapsto \nu_1}$ as shorthand for the set $\{x_1  \mid x_1 = \nu_1\}$.

We can propagate this set of initial states through statements, for example:
\begin{align*}
&\hoare{x_1 \mapsto \nu_1}\\
  &\negate(x_1) \\
  &\hoare{x_1 \mapsto \neg \nu_1}
\end{align*}
The Hoare-style annotation says that the final value of $x_1$ is the negation of its original value.

Similarly, here's an example of propagating all possible initial states through a conditional statement:
\begin{align*}
  &\hoare{x_1 \mapsto \nu_1 , x_2 \mapsto \nu_2}\\
  &\ifflip{x_1}{x_2} \\
  &\hoare{x_1 \mapsto \nu_1 , x_2 \mapsto \nu_1 \oplus \nu_2}
\end{align*}
Notice how the post-condition captures the semantics of the conditional via an XOR---a true $x_1$ forces $x_2$ to flip its initial value, $\nu_2$.

The $\read$ statement, since it is non-deterministic (we don't know what the user will input), is handled by introducing fresh variables---a standard trick in symbolic execution and abstract interpretation. For example, in our first example, the $\read$ statement generates the fresh variable $\nu_3$ (highlighted).
\begin{align*}
  &\hoare{x_1 \mapsto \nu_1 , x_2 \mapsto \nu_2}\\
  &\negate(x_1) \\
  &\hoare{x_1 \mapsto \neg\nu_1 , x_2 \mapsto \nu_2}\\
  &\read(x_2) \\
  &\hoare{x_1 \mapsto \neg\nu_1 , x_2 \mapsto \fresh{\nu_3}}\\
  &\ifflip{x_2}{x_1} \\
  &\hoare{x_1 \mapsto \fresh{\nu_3} \oplus \neg\nu_1 , x_2 \mapsto \fresh{\nu_3}}
\end{align*}

For our little programming language, the annotation may grow linearly with the size of the program---as fresh variables are introduced and conditionals propagate them.
Traditionally, we want our static analysis to be sound, that is, to overapproximate the set of reachable states at each program location. But what if we relax the soundness constraint a little bit? What if, say, we're OK with it being unsound with some small probability?

\subsection{Probabilistically Sound Analysis}
\paragraph{Predicate equivalence}
Notice how our language is peculiar in that every annotation we can derive will be a sequence of XORs. If we normalize out the negations ($\neg$), all our predicates will be of the form
$$\hoare{\ldots, x_i \mapsto c \oplus \bigoplus_{\nu \in S} \nu, \ldots}$$
where $c$ is the constant 1 or 0 ($\true$ or $\false$) and $S$ is some subset of the $\nu$ variables used in our annotations. For example, in our last annotation in the example above, we have
$$x_1 \mapsto 1 \oplus \nu_3 \oplus \nu_1$$

Suppose we are given two expressions of the form above,
$$c_1 \oplus \bigoplus_{\nu \in S_1} \nu \ \ \text{and} \ \  c_2 \oplus \bigoplus_{\nu \in S_2} \nu$$
How can we check if they are equivalent? 
We can call a SAT solver, but this is overkill for such simple formulas.
Instead, we can syntactically check that the constants agree ($c_1 = c_2$) and the variable sets agree ($S_1 = S_2$).

But we can do even better: we can evaluate both expressions by instantiating the $\nu$ variables with values drawn uniformly at random.
If both expressions evaluate to the same value, we deem them equivalent.
Of course, this incurs a non-trivial probability of one-sided error---mistakenly saying that the two formulas are equivalent.
To boost the probability of success, we can instead draw many values by simply treating the $\nu$ variables as \emph{bitstrings} over some number of bits $k$, where $\oplus$ is bitwise XOR.
This diminishes the probability of failure exponentially in $k$.
So we can pick $k=64$ or $k=128$ and be extremely confident in our result.

\paragraph{A randomized static analysis}
This leads us to an interesting observation about our language: We can derive an annotation by using random bitstrings as stand-ins for the $\nu$ variables.
So, now each annotation will be of the form $\hoare{x_1 \mapsto b_1 , x_2 \mapsto b_2 ,...}$,
where each $b_i$ is a bitstring randomly sampled at the start of the analysis.
We interpret this annotation as denoting the set of program states $$\{(x_1,\ldots,x_n) \mid \forall i, j.\;\; b_i = b_j \Rightarrow x_i = x_j\}$$
Notice the one-directional implication.

This abstraction has a number of interesting consequences: 
\begin{enumerate}
\item we no longer need to maintain arbitrarily long expressions, just a single constant (bitstring) per variable $x_i$;
\item propagating the state forward is a simple, bit-level operation (bit flips and XORs);
\item there's a very small probability that our annotation does not overapproximate the set of reachable states.
\end{enumerate}

Let's consider the simple program
\begin{align*}
  &\negate(x_1) \\
  &\negate(x_2)
\end{align*}
and annotate it using our randomized approach.
For each $x_i$, we will pick a random bitstring $b_i$ (chosen uniformly at random).
So we get the following results:
\begin{align*}
  &\phoare{x_1 \mapsto b_1,  x_2 \mapsto b_2}{\delta}\\
  &\negate(x_1) \\
  &\phoare{x_1 \mapsto \neg b_1 , x_2 \mapsto b_2}{\delta}\\
  &\negate(x_2)\\
  &\phoare{x_1 \mapsto \neg b_1 , x_2 \mapsto \neg b_2}{\delta}
\end{align*}
Notice how we add a $\delta$ to our annotation; this is an upper bound on the probability of failure of the equality tests that we can make on the state.
Let's take the precondition:
$$\phoare{x_1 \mapsto b_1 , x_2 \mapsto b_2}{\delta}$$
and say we check if $x_1 = x_2$ in \emph{all} initial states. 
Since $b_1$ and $b_2$ are $k$-bitstrings, drawn uniformly at random, this means that with a probability of $$\delta \leq 1/{2^k}$$ we will incorrectly conclude that $x_1 = x_2$.

\paragraph{The moral of the story}
In summary, we have simplified our static analysis by accepting a small probability of failure: our static analysis implicitly tracks equality predicates across variables and with a low probability may say that $x_i = x_j$ when that is not always true. 
But this paper is about optimizing quantum circuits, so what was that all about?

\subsection{Quantum Circuits}
A quantum circuit operates on a set of quantum bits (\emph{qubits}), which, like our Boolean variables, take values in $\{0,1\}$.
Notationally, if we have a circuit over a single qubit $\q{q_1}$, then its state is either $\ket{0}$ or $\ket{1}$.
But it can also be a \emph{superposition} of both, written as $$\alpha\ket{0} + \beta\ket{1}$$
where $\alpha, \beta \in \mathbb{C}$ are amplitudes.
We call $\ket{0}$ and $\ket{1}$ the \emph{basis states} of a single qubit, and the amplitudes are complex numbers that satisfy $|\alpha|^2 + |\beta|^2 = 1$.
This generalizes to \emph{basis states} for more than one qubit, e.g., we can have the state $\alpha\ket{000} + \beta\ket{111}$.
A quantum circuit transforms these states.

We will consider a simple quantum circuit with the following three instructions (gates): $\X$ (NOT), $\CX$ (controlled NOT), and $\Had$ (Hadamard).
Interestingly, these have a direct correspondence to our classical language above. The $\X$ gate is $\negate$ and the $\CX$ gate is the conditional.
The $\Had$ (Hadamard) gate is the one that creates superposition: it splits a qubit's value into both $0$ and $1$ simultaneously. We will conservatively treat it as a non-deterministic operation, like $\read$.

\begin{center}
\begin{tabular}{ll}
\toprule
\textbf{Classical} & \textbf{Quantum} \\
\midrule
$\negate(x_i)$ & $\X\ \q{q_i}$ \\
$\ifflip{x_i}{x_j}$ & $\CX\ \q{q_i}\ \q{q_j}$ \\
$\read(x_i)$ & $\Had\ \q{q_i}$ \\
\bottomrule
\end{tabular}
\end{center}

What this implies is that we can apply our static analysis that we defined above to quantum circuits.
In fact, take any of our examples above, replace the classical instructions with their quantum counterparts, and you have a sound static analysis that does not track amplitudes, just reachable basis states with \emph{nonzero} amplitude.

Consider the following simple example circuit with a single qubit $\q{q_1}$
\begin{align*}
&\hoare{\q{q_1} \mapsto \nu_1}\\
  &\X\ \q{q_1} \\
  &\hoare{\q{q_1} \mapsto \neg \nu_1}
\end{align*}
This captures the transition relation of this quantum gate: if $\q{q_1}$ starts in some state $|\nu_1\rangle$,
then after executing the $\X$ statement it ends up in $|\neg \nu_1\rangle$.

Similarly, this example with two qubits demonstrates the behavior of a controlled NOT---if $\q{q_1}$ then flip $\q{q_2}$.
\begin{align*}
  &\hoare{\q{q_1} \mapsto \nu_1 , \q{q_2} \mapsto \nu_2}\\
  &\CX\ \q{q_1}\ \q{q_2} \\
  &\hoare{\q{q_1} \mapsto \nu_1 , \q{q_2} \mapsto \nu_1 \oplus \nu_2}
\end{align*}

\paragraph{Guarantees for quantum circuits} Formally, this analysis overapproximates the transition relation of the circuit---the pairs of input--output basis states related with nonzero amplitude. We can also apply our randomized static analysis with the same effect as in the classical setting: it will overapproximate the transition relation, but with a small probability it may fail.

Consider the $\CX$ above and let's run the randomized analysis, where we sample the initial values $b_i$ as bitstrings uniformly at random:
\begin{align*}
  &\phoare{\q{q_1} \mapsto b_1 , \q{q_2} \mapsto b_2}{\delta}\\
  &\CX\ \q{q_1}\ \q{q_2} \\
  &\phoare{\q{q_1} \mapsto b_1 , \q{q_2} \mapsto b_1 \oplus b_2}{\delta}
\end{align*}
Naturally, $\q{q_1}$ and $\q{q_2}$ are not guaranteed to be the same at the end of this circuit, but there's a small probability $\delta$ that our analysis may deduce that $\q{q_1}=\q{q_2}$ for all input basis states. This occurs precisely when $b_1 = b_1 \oplus b_2$, i.e., when $b_2 = 0^k$, which happens with probability $2^{-k}$---exceedingly rarely as we increase the number of bits $k$.

\subsection{The Phase-Folding Optimization}
There is another gate that we did not include above: the phase-rotation gate $\Rz{\theta}$, which multiplies the $\ket{1}$ component of a qubit's amplitude by $e^{i\theta}$, where $\theta \in \mathbb{R}$.
This gate only modifies the amplitude of a quantum state, not its basis states.
So if we include it in our analysis, we simply treat it as a no-op, since our analysis does not track amplitudes.

Non-Clifford phase rotations are expensive in fault-tolerant quantum computing, as they are typically synthesized using costly magic-state resources~\cite{bravyi2005universal,bravyi2012magic,gidney2024cultivation}.
In particular, the $T$ gate ($= \Rz{\pi/4}$) dominates the resource budget.
We therefore want to minimize the number of rotation gates in a circuit.

\paragraph{Merging adjacent rotations}
A simple optimization is to merge two rotations on the same qubit into one.
For example:
\begin{align*}
  &\Rz{\theta_1}\ \q{q_1} \\
  &\Rz{\theta_2}\ \q{q_1}
\end{align*}
can be replaced by a single $\Rz{\theta_1 + \theta_2}\ \q{q_1}$, since rotation angles simply accumulate.

\paragraph{Merging distant rotations}
But we can do better: we can merge two rotations that are far apart in the circuit, even on \emph{different} qubits, as long as they always \emph{see} the same basis states.
Consider the following circuit, shown as a sequence of gates on the left and a circuit diagram on the right:
\begin{center}
\begin{minipage}{0.35\linewidth}
\begin{align*}
  &\Rz{\theta_1}\ \q{q_1} \\
  &\CX\ \q{q_1}\ \q{q_2} \\
  &\CX\ \q{q_2}\ \q{q_1} \\
  &\CX\ \q{q_1}\ \q{q_2} \\
  &\Rz{\theta_2}\ \q{q_2}
\end{align*}
\end{minipage}%
\begin{minipage}{0.6\linewidth}
\centering
\begin{quantikz}[column sep=0.3cm]
  \lstick{$\q{q_1}$}
    & \gate{\Rz{\theta_1}} & \ctrl{1} & \targ{}   & \ctrl{1} & \qw & \qw \\
  \lstick{$\q{q_2}$}
    & \qw                  & \targ{}  & \ctrl{-1} & \targ{}  & \gate{\Rz{\theta_2}} & \qw
\end{quantikz}
\end{minipage}
\end{center}
Here, the three $\CX$ gates swap the values of $\q{q_1}$ and $\q{q_2}$.
The first rotation acts on $\q{q_1}$, and the second on $\q{q_2}$---a different qubit.
Yet because the swap (the three $\CX$'s) moves the original value of $\q{q_1}$ into $\q{q_2}$, both rotations always see the same value, and we can merge them into a single $\Rz{\theta_1 + \theta_2}\ \q{q_2}$.

Our randomized analysis detects this automatically. We sample the initial values of each qubit as a uniformly random bitstring, $\q{q_1} \mapsto b_1$ and $\q{q_2} \mapsto b_2$, and propagate them through the circuit.
Below, we illustrate how our analysis propagates the initial bitstrings $b_1$ and $b_2$ through the circuit, showing the values of $\q{q_1}$ and $\q{q_2}$ at each point in the circuit.

\begin{center}
\begin{quantikz}[column sep=0.3cm]
  \lstick{$\q{q_1}$}
    & \push{\color{hoareblue}b_1} & \gate{\Rz{\theta_1}}
    & \push{\color{hoareblue}b_1} & \ctrl{1}
    & \push{\color{hoareblue}b_1} & \targ{}
    & \push{\color{hoareblue}b_2} & \ctrl{1}
    & \push{\color{hoareblue}b_2} & \qw
    & \push{\color{hoareblue}b_2} & \qw \\
  \lstick{$\q{q_2}$}
    & \push{\color{hoareblue}b_2} & \qw
    & \push{\color{hoareblue}b_2} & \targ{}
    & \push{\color{hoareblue}b_1 \oplus b_2} & \ctrl{-1}
    & \push{\color{hoareblue}b_1 \oplus b_2} & \targ{}
    & \push{\color{hoareblue}b_1} & \gate{\Rz{\theta_2}}
    & \push{\color{hoareblue}b_1} & \qw
\end{quantikz}
\end{center}

Notice that, at the first $\Rz{\theta_1}$, qubit $\q{q_1}$ holds $b_1$; at the second, $\q{q_2}$ also holds $b_1$.
So, no matter what values of $b_1$ and $b_2$ we start with, we will be able to deduce that merging the two rotations is sound.
This, however, is not always true, and our analysis may incur a tiny failure probability $\delta$ that vanishes as the width of the bitstring $b_i$ grows.

To make this example concrete, take $k = 3$ and draw $b_1 = 101$, $b_2 = 011$. The same propagation produces:
\begin{center}
\begin{quantikz}[column sep=0.3cm]
  \lstick{$\q{q_1}$}
    & \push{\color{hoareblue}101} & \gate{\Rz{\theta_1}}
    & \push{\color{hoareblue}101} & \ctrl{1}
    & \push{\color{hoareblue}101} & \targ{}
    & \push{\color{hoareblue}011} & \ctrl{1}
    & \push{\color{hoareblue}011} & \qw
    & \push{\color{hoareblue}011} & \qw \\
  \lstick{$\q{q_2}$}
    & \push{\color{hoareblue}011} & \qw
    & \push{\color{hoareblue}011} & \targ{}
    & \push{\color{hoareblue}110} & \ctrl{-1}
    & \push{\color{hoareblue}110} & \targ{}
    & \push{\color{hoareblue}101} & \gate{\Rz{\theta_2}}
    & \push{\color{hoareblue}101} & \qw
\end{quantikz}
\end{center}
At the first $\Rz{\theta_1}$, $\q{q_1}$ holds $101$; at the second, $\q{q_2}$ also holds $101$. The bitstrings agree, so the analysis concludes that the rotations can be merged.

\paragraph{When folding should not happen}
Not every pair of rotations can be merged.
Consider the circuit
\begin{center}
\begin{minipage}{0.6\linewidth}
\centering
\begin{quantikz}[column sep=0.3cm]
  \lstick{$\q{q}$} & \gate{\Rz{\theta_1}} & \gate{\Had} & \gate{\Rz{\theta_2}} & \qw
\end{quantikz}
\end{minipage}
\end{center}
with two rotations on the same qubit with a Hadamard in between.
The Hadamard mixes the $Z$-basis, so on different input paths $\q{q}$ arrives
at the second rotation holding different values---the two rotations do not see the
same bit, and merging them would change the circuit's behavior.

Applying our randomized analysis generates the following annotation:
\begin{center}
\begin{quantikz}[column sep=0.3cm]
  \lstick{$\q{q}$}
    & \push{\color{hoareblue}b}  & \gate{\Rz{\theta_1}}
    & \push{\color{hoareblue}b}  & \gate{\Had}
    & \push{\color{hoareblue}b'} & \gate{\Rz{\theta_2}}
    & \push{\color{hoareblue}b'} & \qw
\end{quantikz}
\end{center}
Qubit $\q{q}$ starts with a uniformly random bitstring $b$, which is carried
unchanged through $\Rz{\theta_1}$; the Hadamard then draws a fresh, independent
bitstring $b'$.
Since $b$ and $b'$ are independent uniform draws, $b \neq b'$ with probability
$1 - 2^{-k}$, so the two rotation sites carry different bitstrings and the
analysis correctly declines to merge them with a high probability.

\paragraph{A fast algorithm}
In what follows, we formalize our approach and show how it yields a phase-folding algorithm that is \emph{linear} in the size of the circuit for a constant bitstring width $k$.
The algorithm walks through the circuit from beginning to end, maintaining the results of the probabilistic analysis.
Since each qubit maps to a constant-width bitstring, we can check whether we have seen a matching qubit state before, using a constant-time hash lookup---a technique akin to Zobrist hashing~\cite{zobrist1970new}.
This, as we shall see in \cref{sec:evaluation}, results in a fast and scalable optimization tool.

\section{Phase Folding, Formally}
\label{sec:formal}
We now formally define quantum circuits, a relational semantics, and the phase-folding problem.

\subsection{Quantum Circuits and their Semantics}
\paragraph{Quantum circuits}
We define a quantum circuit $\circuit$ as a sequence of gates
$g_1, g_2, \ldots, g_m$ over a set of qubits $\q{q_1},\ldots, \q{q_n}$.
We restrict our gates to the set $$\{\CX, \Had, \X, \Rz{\theta}\}$$ 
where
$\theta \in \mathbb{R}$.
The $T$ gate is $\Rz{\pi/4}$, the $S$ gate is $\Rz{\pi/2}$, and the $Z$ gate is $\Rz{\pi}$.

\paragraph{Weighted-relation semantics}
To reason about circuits for our purposes, we will give our circuits a semantics
as a weighted relation on \emph{basis states} in $\{0,1\}^n$.
Think of a \emph{quantum state} as a function $\psi \colon \{0,1\}^n \to \mathbb{C}$
that assigns a complex amplitude to each basis state.
A circuit $\circuit$ then induces a relation
\[
  \sem{\circuit} \;\colon\; \{0,1\}^n \times \{0,1\}^n \;\to\; \mathbb{C}
\]
where $\sem{\circuit}(x, x')$ is the total amplitude with which
$\circuit$ maps input basis state $x$ to output basis state $x'$.
We write $x_\q{q}$ for the value of qubit $\q{q}$ in basis state $x$,
and $x[\q{q} \mapsto v]$ for the basis state that agrees with $x$
everywhere except qubit $\q{q}$, which is set to $v$.
For example, for the basis state $x=01$ over qubits $\q{q_1}, \q{q_2}$, we have $x_{\q{q_1}} = 0$, $x_{\q{q_2}} = 1$, and $x[\q{q_1} \mapsto 1] = 11$.

The semantics of each gate is given in \Cref{fig:gate-semantics}.

\begin{example}
Consider the gate $\X\ \q{q_1}$ and suppose we have two qubits in the circuit.
For input state $x = 00$, the gate flips $\q{q_1}$, so 
$\sem{\X\ \q{q_1}}(00, 10) = 1$ and  $\sem{\X\ \q{q_1}}(00, x') = 0$ for all $x' \neq 10$.
Similarly, $\sem{\X\ \q{q_1}}(01, 11) = 1$.
\end{example}

\begin{example}
Consider $\Had\ \q{q_1}$ on two qubits.
For input $x = 00$, the gate fans out to both values of $\q{q_1}$:
$\sem{\Had\ \q{q_1}}(00, 00) = \tfrac{1}{\sqrt{2}}$ and
$\sem{\Had\ \q{q_1}}(00, 10) = \tfrac{1}{\sqrt{2}}$.
\end{example}

\paragraph{Composition} Composing two circuits means summing over intermediate basis states:
\[
  \sem{\circuit_1 \,;\, \circuit_2}(x, x')
  \;=\;
  \sum_{x'' \in \{0,1\}^n}
    \sem{\circuit_1}(x, x'') \cdot \sem{\circuit_2}(x'', x')
\]

\begin{figure}[t]
\centering\small
\begin{minipage}{0.95\columnwidth}
\begin{minipage}[t]{0.48\columnwidth}
\begin{align*}
  \sem{\Rz{\theta}\ \q{q}}(x, x') &=
    \begin{cases}
      e^{i\theta \cdot x_\q{q}} & \text{if } x' = x \\
      0 & \text{otherwise}
    \end{cases}
\\[6pt]
  \sem{\CX\ \q{c}\ \q{t}}(x, x') &=
    \begin{cases}
      1 & \text{if } x' = x[\q{t} \mapsto x_\q{t} \oplus x_\q{c}] \\
      0 & \text{otherwise}
    \end{cases}
\end{align*}
\end{minipage}
\hfill
\begin{minipage}[t]{0.48\columnwidth}
\begin{align*}
  \sem{\X\ \q{q}}(x, x') &=
    \begin{cases}
      1 & \text{if } x' = x[\q{q} \mapsto 1 - x_\q{q}] \\
      0 & \text{otherwise}
    \end{cases}
\\[6pt]
  \sem{\Had\ \q{q}}(x, x') &=
    \begin{cases}
      \dfrac{(-1)^{x_\q{q} \cdot x'_\q{q}}}{\sqrt{2}}
        & \text{if } x'_\q{r} = x_\q{r} \text{ for all } \q{r} \neq \q{q} \\[6pt]
      0 & \text{otherwise}
    \end{cases}
\end{align*}
\end{minipage}
\end{minipage}
\caption{Gate semantics as weighted input/output relations.
Each gate $g$ denotes a function
$\sem{g} \colon \{0,1\}^n \times \{0,1\}^n \to \mathbb{C}$
assigning to each pair of basis states the amplitude with which $g$ maps the first to the second.
Sequential composition multiplies weights and sums over intermediate basis states.}
\label{fig:gate-semantics}
\end{figure}

\paragraph{Circuit equivalence}
Two circuits are \emph{equivalent}, written $\circuit_1 \equiv \circuit_2$,
when $\sem{\circuit_1} = \sem{\circuit_2}$---i.e., they induce the same weighted relation on basis states.

\subsection{The Phase-Folding Conditions}
We are now ready to state the condition under which two $Z$-rotations can be merged.
Consider a circuit of the form
\[
  \circuit_1 \;;\; \Rz{\theta}\ \q{q} \;;\;
  \circuit_2 \;;\; \Rz{\theta'}\ \q{q'} \;;\; \circuit_3
\]
The following theorem gives a sufficient condition for merging the two rotations into $\Rz{\theta + \theta'}\ \q{q'}$, even though they are on different qubits and are embedded within a larger circuit.
(While phase folding is well-studied, we have not found a general formulation of phase folding in the literature where rotations appear within a larger circuit, and so we formalize it here, as it is the basis for our optimization algorithm.)

\begin{theorem}[Phase folding]
\label{thm:phase-merge}
Consider the following pair of circuits:
\begin{align*}
  \circuit_{\mathrm{left}}
    \;&=\; \circuit_1 \;;\; \Rz{\theta}\ \q{q} \;;\; \circuit_2 \;;\; \Rz{\theta'}\ \q{q'} \;;\; \circuit_3 \\
  \circuit_{\mathrm{right}}
    \;&=\; \circuit_1 \;;\; \circuit_2 \;;\; \Rz{\theta + \theta'}\ \q{q'} \;;\; \circuit_3 .
\end{align*}
If for all basis states $x_0,x,x'$ with
$\sem{\circuit_1}(x_0,x) \neq 0$ and
$\sem{\circuit_2}(x,x') \neq 0$ we have
$x_\q{q} = x'_\q{q'}$,
then $\circuit_{\mathrm{left}} \equiv \circuit_{\mathrm{right}}$.
\end{theorem}
\begin{proof}
By the composition rule,
\begin{align*}
  \sem{\circuit_{\mathrm{left}}}(y,y')
  &=
  \sum_{a,b,c,d}
    \sem{\circuit_1}(y,a)
    \cdot \sem{\Rz{\theta}\ \q{q}}(a,b)
    \cdot \sem{\circuit_2}(b,c)
    \cdot \sem{\Rz{\theta'}\ \q{q'}}(c,d)
    \cdot \sem{\circuit_3}(d,y') .
\end{align*}
 $\sem{\Rz{\theta}\ \q{q}}(a,b)$ is zero unless
$a=b$, and in the case $a=b$ it contributes the phase
$e^{i\theta a_\q{q}}$.
Likewise, $\sem{\Rz{\theta'}\ \q{q'}}(c,d)$ is zero unless $c=d$, and then
contributes $e^{i\theta' c_\q{q'}}$.
Thus the sum reduces to
\[
  \sem{\circuit_{\mathrm{left}}}(y,y')
  =
  \sum_{x,x'}
    \sem{\circuit_1}(y,x)
    \cdot e^{i\theta x_\q{q}}
    \cdot \sem{\circuit_2}(x,x')
    \cdot e^{i\theta' x'_\q{q'}}
    \cdot \sem{\circuit_3}(x',y') .
\]
Consider any nonzero summand in this last expression.
Then $\sem{\circuit_1}(y,x) \neq 0$ and
$\sem{\circuit_2}(x,x') \neq 0$, so the theorem premise applies with
$x_0 = y$ and gives $x_\q{q} = x'_\q{q'}$.
Therefore the two phase factors in that summand combine as
\[
  e^{i\theta x_\q{q}} e^{i\theta' x'_\q{q'}}
  =
  e^{i(\theta+\theta')x'_\q{q'}} .
\]
Zero summands contribute nothing, so substituting this equality into the sum
gives
\[
  \sem{\circuit_{\mathrm{left}}}(y,y')
  =
  \sum_{x,x'}
    \sem{\circuit_1}(y,x)
    \cdot \sem{\circuit_2}(x,x')
    \cdot e^{i(\theta+\theta')x'_\q{q'}}
    \cdot \sem{\circuit_3}(x',y') ,
\]
which is exactly $\sem{\circuit_{\mathrm{right}}}(y,y')$.
\end{proof}

\section{The Parity Equivalence Problem}
\label{sec:analysis}
Our randomized static analysis is based on the problem of checking equivalence of Boolean formulas of a certain structure: \emph{affine functions}, or \emph{parities}.
This section formalizes this simple class of formulas and states \Cref{lem:parity-eq}, the key lemma that makes our randomized analysis work.

\paragraph{Parities}
Let $V = \{v_0, v_1, v_2, \ldots\}$ be a set of Boolean variables.
A \emph{parity} $p$ is a function over Boolean variables that has the following form:
\[
  p \;=\; c \oplus \bigoplus_{v \in S} v
\]
where $S \subseteq V$ is a finite set, $c \in \{0, 1\}$ is a constant, and $\oplus$ is XOR.
Given a \emph{valuation} $\alpha \colon V \to \{0,1\}$, the value of $p$ under $\alpha$ is
\[
  \alpha(p) \;=\; c \oplus \bigoplus_{v \in S} \alpha(v).
\]
Two parities are \emph{equal} if they agree on every valuation; equivalently, if they share the same constant $c$ and variable set $S$.
Checking equality symbolically reduces to comparing the sets $S$ and constants $c$.

\paragraph{Parity equality via sampling}
A simpler equality test is to draw a uniformly random bitstring of width $k$ for each variable and evaluate the two parities on bitstrings.
We extend a $k$-bitstring valuation $\alpha \colon V \to \{0,1\}^k$ to parities by treating the constant $c \in \{0,1\}$ as the bitstring $c^k$ (i.e., $0^k$ or $1^k$) and performing all XORs bitwise:
\[
  \alpha(p) \;=\; c^k \;\oplus\; \bigoplus_{v \in S} \alpha(v) \;\in\; \{0,1\}^k.
\]
\begin{example}
For example, with $k = 4$, $p_1 = v_1 \oplus v_2$, and $p_2 = v_1 \oplus v_3$, if we draw the following valuation, 
\[
  \alpha(v_1) = 1011, \quad \alpha(v_2) = 1100, \quad \alpha(v_3) = 0010,
\]
then $\alpha(p_1) = 1011 \oplus 1100 = 0111$ and $\alpha(p_2) = 1011 \oplus 0010 = 1001$, so we correctly declare $p_1 \neq p_2$.
\end{example}
The lemma below bounds the probability that two unequal parities are falsely declared equal by this test.

\begin{lemma}[Randomized parity equality]
\label{lem:parity-eq}
Let $p_1$ and $p_2$ be two parities over $V$, and let $k \geq 1$ be a bitstring width.
Let $\alpha \colon V \to \{0,1\}^k$ be a uniformly random valuation
(each variable independently uniform on $\{0,1\}^k$). Then:
\begin{enumerate}
  \item If $p_1 = p_2$, then $\alpha(p_1) = \alpha(p_2)$ with probability $1$.
  \item If $p_1 \neq p_2$, then $\alpha(p_1) = \alpha(p_2)$ with probability at most $2^{-k}$.
\end{enumerate}
\end{lemma}
\begin{proof}
\emph{Part~(1).}
If $p_1 = p_2$, they agree on every valuation, so $\alpha(p_1) = \alpha(p_2)$ for every $\alpha$, with probability $1$.

\emph{Part~(2).}
Writing $p_j = c_j \oplus \bigoplus_{v \in S_j} v$ for $j \in \{1,2\}$, define
\[
  p \;=\; p_1 \oplus p_2 \;=\; c_p \oplus \bigoplus_{v \in S_p} v,
  \qquad
  c_p \;=\; c_1 \oplus c_2,
  \quad
  S_p \;=\; S_1 \triangle S_2,
\]
where $\triangle$ is symmetric difference.
Since $p_1 \neq p_2$, we have $p \neq 0$, so either $S_p \neq \emptyset$ or $c_p = 1$.

\emph{Case 1: $S_p = \emptyset$.} Then $c_p = 1$ and $\alpha(p) = 1^k$, so $\alpha(p_1) \neq \alpha(p_2)$ deterministically. The collision probability is $0 \leq 2^{-k}$.

\emph{Case 2: $S_p \neq \emptyset$.} Pick any $v \in S_p$ and write
\[
  \alpha(p)
  \;=\;
  \alpha(v) \;\oplus\; \underbrace{c_p^k \oplus \bigoplus_{v' \in S_p \setminus \{v\}} \alpha(v')}_{=:\; c'},
\]
where $c'$ is determined by the remaining variables and the constant $c_p$.
Since $\alpha(v)$ is uniformly distributed over $\{0,1\}^k$ and is independent of $c'$, $\alpha(p) = \alpha(v) \oplus c'$ is also uniform on $\{0,1\}^k$. Therefore
\[
  \Pr\bigl[\alpha(p_1) = \alpha(p_2)\bigr]
  \;=\;
  \Pr\bigl[\alpha(p) = 0^k\bigr]
  \;=\;
  2^{-k}.
\]
In either case, the probability that $\alpha(p_1) = \alpha(p_2)$ is at most $2^{-k}$. \qedhere
\end{proof}

\begin{example}
Let $p_1 = v_0 \oplus v_1$ and $p_2 = v_0 \oplus v_2$, so $p_1 \neq p_2$.
Draw $k = 4$ bitstrings uniformly at random:
\[
  \alpha(v_0) = 1011, \quad \alpha(v_1) = 0110, \quad \alpha(v_2) = 0110.
\]
Then $\alpha(p_1) = 1011 \oplus 0110 = 1101$ and $\alpha(p_2) = 1011 \oplus 0110 = 1101$.
The two unequal parities collide here---a false positive that occurs with probability $2^{-4} = 1/16$.
With $k = 128$ such collisions are astronomically rare.
\end{example}

\section{Randomized Static Analysis and Phase Folding}
\label{sec:randomized}

We are now ready to present our randomized static analysis and show how it can be used to
realize a very fast phase-folding optimization algorithm.

\paragraph{What we are abstracting}
Our goal is to design an abstraction of the relational circuit semantics
$\sem{\circuit}$.
The abstraction should overapproximate the pairs of basis states $x,x'$ such that $\sem{\circuit}(x,x') \neq 0$ and be precise enough to detect when two rotations can be folded,
yet cheap enough to compute in a single linear pass. Recall from
\Cref{thm:phase-merge} that the phase-folding side condition is a statement about
$\sem{\circuit}$: every input--output pair of basis states related by the circuit
must agree on the values of two specific qubits.

\paragraph{The abstraction}
We build an abstract domain in which each qubit is mapped to a $k$-bitstring,
initialized uniformly at random and updated by a small set of bitwise \emph{transfer
functions}, one per gate. Throughout the circuit, we treat the equality of two
qubits' bitstrings as evidence that the qubits hold the same value in every
reachable state. In other words, our abstraction implicitly tracks equalities between qubits.

\paragraph{Soundness with high probability}
This is a one-sided guarantee: we do not claim to detect every true equality, but for each equality query we rely on, if the two queried qubit values \emph{can} disagree on some reachable transition, the analysis spuriously assigns them the same bitstring with probability at most $2^{-k}$.
The parity equivalence lemma (\Cref{lem:parity-eq}) is the key ingredient: each basis state reachable at any program point is a parity of the initial qubit values and the Hadamard outcomes, so equality between two qubits' values reduces to equality of parities, which we test by sampling.

\subsection{Randomized Static Analysis}
We are now ready to define our randomized static analysis, which propagates a randomized abstract state through the circuit, one gate at a time, as illustrated informally in \Cref{sec:overview}.

\paragraph{Randomized abstract state}
Fix a width $k \geq 1$ for bitstrings.
A \emph{randomized abstract state} $\sigma$ is a map that assigns to each qubit $\q{q}$ a
bitstring $\sigma(\q{q}) \in \{0,1\}^k$.
We are going to treat $\sigma$ as an abstract quantum state,
and propagate it down the circuit, one gate at a time, by abstractly interpreting the semantics of the circuit to be transformations of such maps.
Intuitively, $\sigma$ will be an approximation of the relational semantics $\sem{\circuit}$.

We create the initial state $\sigma_0$ by drawing each qubit $\q{q}$'s bitstring independently and uniformly:
\[
  \sigma_0(\q{q}) \;\sim\; \mathrm{Uniform}(\{0,1\}^k),
\]
where $\mathrm{Uniform}(\{0,1\}^k)$ is the uniform distribution over bitstrings of length $k$.

\paragraph{Randomized abstract transfer functions}
Each gate updates $\sigma$ into $\sigma'$ using the following rules. Note that each rule updates the bitstring of at most one qubit, while keeping all others the same:
\begin{align*}
  \rasem{\CX\ \q{c}\ \q{t}}(\sigma) &\colon \quad \sigma'(\q{t}) \;=\; \sigma(\q{t}) \oplus \sigma(\q{c}) \\
  \rasem{\X\ \q{q}}(\sigma) &\colon \quad \sigma'(\q{q}) \;=\; \sigma(\q{q}) \oplus \mathbf{1}_k \\
  \rasem{\Had\ \q{q}}(\sigma) &\colon \quad \sigma'(\q{q}) \;\sim\; \mathrm{Uniform}(\{0,1\}^k) \\
  \rasem{\Rz{\theta}\ \q{q}}(\sigma) &\colon \quad \sigma' \;=\; \sigma
\end{align*}
Note that $\mathbf{1}_k = (1,\ldots,1)$ and $\oplus$ is bitwise XOR.
The $\Had$ rule draws a fresh uniform bitstring, independently of all prior draws,
because the Hadamard gate severs the linear relationship between a qubit and the inputs.
Every update is a single word-level operation.

\begin{example}
Take $k = 3$ and a circuit on two qubits $\q{q_1}, \q{q_2}$.
Suppose the initial bitstrings, drawn at random, are
\[
  \sigma_0(\q{q_1}) = 101, \qquad \sigma_0(\q{q_2}) = 011.
\]
We propagate $\sigma_0$ through three gates, getting the following abstract states along the circuit (in blue):
\begin{align*}
  &\phoare{\q{q_1} \mapsto 101,\; \q{q_2} \mapsto 011}{}\\
  &\CX\ \q{q_1}\ \q{q_2} \\
  &\phoare{\q{q_1} \mapsto 101,\; \q{q_2} \mapsto 101 \oplus 011 = 110}{}\\
  &\X\ \q{q_1} \\
  &\phoare{\q{q_1} \mapsto 101 \oplus 111 = 010,\; \q{q_2} \mapsto 110}{}\\
  &\Had\ \q{q_2} \\
  &\phoare{\q{q_1} \mapsto 010,\; \q{q_2} \mapsto \fresh{100}}{}
\end{align*}
The bitstring \fresh{100} is a freshly sampled uniform bitstring (shaded).
After the Hadamard, $\q{q_2}$'s bitstring is independent of $\q{q_1}$'s,
reflecting the fact that the gate breaks the linear relationship between $\q{q_2}$ and the inputs.
\end{example}

\subsection{Soundness of the Analysis}
\label{sec:soundness}
The information our analysis extracts from a circuit is \emph{equalities between
qubits}: whenever two qubits carry the same bitstring, we treat
them as guaranteed to hold the same value in every reachable state. Soundness is
therefore a statement about these equalities---we want every equality the randomized
analysis detects to reflect a genuine semantic equality in $\sem{\circuit}$.

The following theorem formalizes soundness of our randomized analysis.
Notice that this is a one-sided guarantee:
If two qubits are always the same in all basis states across a transition, then we don't care if our analysis misses this fact. But if they are not always the same, we want our analysis to detect that with high probability.

\begin{theorem}
  [Soundness of the randomized analysis]
  \label{lem:rand-equality}
  Suppose there exist basis states $x,x'$ such that $\sem{\circuit}(x,x') \neq 0$ and $x_\q{q} \neq x'_\q{q'}$.
  Let
  $\sigma = \rasem{\circuit}(\sigma_0)$.
  Then, $$\Pr[\sigma_0(\q{q}) = \sigma(\q{q'})] \leq 2^{-k}.$$
\end{theorem}
Notice that the soundness theorem above relates $\sigma_0$ and $\sigma$---the initial and final abstract states.

\paragraph{Proof idea: Correspondence with symbolic analysis}
The key idea underlying the proof of soundness is that the randomized analysis can be viewed as a random valuation of a symbolic analysis that tracks each qubit's parity, as we illustrated in \Cref{sec:overview}.

We start by defining a symbolic analysis that tracks each qubit's value as a parity of the initial qubit values and the Hadamard outcomes.
A \emph{symbolic abstract state} $\hat{\sigma}$ is a map that assigns to each qubit $\q{q}$ a
parity $\hat{\sigma}(\q{q})$ over a set of Boolean variables $V$.
We create the initial state $\hat{\sigma}_0$ by assigning to each qubit $\q{q}$ its own dedicated variable $v_\q{q} \in V$.

Each gate updates $\hat{\sigma}$ into $\hat{\sigma}'$ using the following rules (notice the one-to-one correspondence with the randomized analysis):
\begin{align*}
  \sasem{\CX\ \q{c}\ \q{t}}(\hat{\sigma}) &\colon \quad \hat{\sigma}'(\q{t}) \;=\; \hat{\sigma}(\q{t}) \oplus \hat{\sigma}(\q{c}) \\
  \sasem{\X\ \q{q}}(\hat{\sigma}) &\colon \quad \hat{\sigma}'(\q{q}) \;=\; \hat{\sigma}(\q{q}) \oplus 1 \\
  \sasem{\Had\ \q{q}}(\hat{\sigma}) &\colon \quad \hat{\sigma}'(\q{q}) \;=\; u \quad \text{where $u \in V$ is a fresh variable} \\
  \sasem{\Rz{\theta}\ \q{q}}(\hat{\sigma}) &\colon \quad \hat{\sigma}' \;=\; \hat{\sigma}
\end{align*}

The parities computed by the analysis mention two kinds of variables: the initial variables $v_\q{q}$, whose values are determined by the input basis state, and the fresh variables introduced by the $\Had$ rule, which we call the \emph{Hadamard variables} and collect in a set $U \subseteq V$.
Given an input $x$ and an assignment $h \colon U \to \{0,1\}$ of the Hadamard variables, let $\alpha_{x,h}$ denote the valuation with
\[
  \alpha_{x,h}(v_\q{q}) = x_\q{q} \quad \text{for every qubit } \q{q},
  \qquad
  \alpha_{x,h}(u) = h(u) \quad \text{for every } u \in U.
\]
Soundness says that the parities computed by the symbolic analysis faithfully describe every input--output pair the circuit can realize: fixing the initial variables to the input, some assignment of the Hadamard variables makes each qubit's parity evaluate to that qubit's output value.

\begin{lemma}[Soundness of the symbolic analysis]
\label{lem:symbolic-soundness}
Let $\hat{\sigma} = \sasem{\circuit}(\hat{\sigma}_0)$, with Hadamard variables $U$.
For all basis states $x, x'$ with $\sem{\circuit}(x, x') \neq 0$,
there exists an assignment $h \colon U \to \{0,1\}$ such that
\[
  \alpha_{x,h}(\hat{\sigma}(\q{q})) = x'_\q{q} \quad \text{for every qubit } \q{q}.
\]
\end{lemma}
We note that the above symbolic analysis and its soundness proof are not particularly novel and are similar to existing analyses, such as that of \citet{amy2025linear}. Nonetheless, we include a proof in \Cref{app:soundness-details} for completeness (see Appendix).

\paragraph{Proof of soundness of the randomized analysis}
With the soundness of the symbolic analysis in hand, we can now prove soundness of the randomized analysis.
The proof below observes that the randomized analysis is equivalent to evaluating the symbolic analysis under a uniformly random valuation of the variables, and then applies \Cref{lem:parity-eq} to bound the probability of a false equality.

\begin{proof}[Proof of \Cref{lem:rand-equality}]
Suppose there exist basis states $x, x'$ with $\sem{\circuit}(x, x') \neq 0$ and $x_\q{q} \neq x'_\q{q'}$.
By \Cref{lem:symbolic-soundness}, there is an assignment $h$ of the Hadamard variables with
$\alpha_{x,h}(\hat{\sigma}(\q{q'})) = x'_\q{q'}$, while $\alpha_{x,h}(\hat{\sigma}_0(\q{q})) = \alpha_{x,h}(v_\q{q}) = x_\q{q}$ by definition.
The two parities $\hat{\sigma}_0(\q{q})$ and $\hat{\sigma}(\q{q'})$ thus disagree under the valuation $\alpha_{x,h}$, so
\[
  \hat{\sigma}_0(\q{q}) \;\neq\; \hat{\sigma}(\q{q'}) \quad \text{as parities.}
\]
Now let $\beta \colon V \to \{0,1\}^k$ be the valuation that assigns to each initial variable $v_\q{q}$ the bitstring drawn for $\sigma_0(\q{q})$, and to each Hadamard variable the fresh bitstring drawn by the $\Had$ rule that introduced it; since all of these draws are independent and uniform, $\beta$ is a uniformly random valuation.
The randomized analysis is precisely the symbolic analysis evaluated under $\beta$: initially $\sigma_0(\q{q}) = \beta(v_\q{q}) = \beta(\hat{\sigma}_0(\q{q}))$ by construction, and evaluation under $\beta$ commutes with each pair of transfer rules---XOR of parities becomes bitwise XOR of bitstrings, the constant $1$ becomes $\mathbf{1}_k$, and a fresh variable evaluates to its fresh draw---so by induction on the circuit, $\sigma(\q{q'}) = \beta(\hat{\sigma}(\q{q'}))$.
Since $\hat{\sigma}_0(\q{q}) \neq \hat{\sigma}(\q{q'})$, \Cref{lem:parity-eq}(2) gives
\[
  \Pr\bigl[\sigma_0(\q{q}) = \sigma(\q{q'})\bigr]
  \;=\;
  \Pr\bigl[\beta(\hat{\sigma}_0(\q{q})) = \beta(\hat{\sigma}(\q{q'}))\bigr]
  \;\leq\; 2^{-k}. \qedhere
\]
\end{proof}

\subsection{Randomized Phase Folding}

We are now ready to connect the randomized analysis to the phase-folding criterion
of \Cref{thm:phase-merge}.
Consider two rotations at distinct program points: if the randomized state assigns
the same bitstring to both rotation qubits at those program points, the rotations
can be merged.

\begin{theorem}[Soundness of randomized phase folding]
\label{thm:rand-soundness}
Consider two rotation sites in a circuit of the form
\[
  \circuit_1 \;;\; \Rz{\theta}\ \q{q} \;;\;
  \circuit_2 \;;\; \Rz{\theta'}\ \q{q'} \;;\; \circuit_3 .
\]
Let the randomized analysis produce states
\[
  \sigma^{(1)} = \rasem{\circuit_1}(\sigma_0)
  \qquad\text{and}\qquad
  \sigma^{(2)} = \rasem{\circuit_2}(\sigma^{(1)}).
\]
If the premise of \Cref{thm:phase-merge} fails---i.e., there exist basis states
$x_0,x,x'$ such that
$\sem{\circuit_1}(x_0,x) \neq 0$,
$\sem{\circuit_2}(x,x') \neq 0$, and
$x_\q{q} \neq x'_\q{q'}$---then
\[
  \Pr\bigl[\sigma^{(1)}(\q{q}) = \sigma^{(2)}(\q{q'})\bigr]
  \;\leq\; 2^{-k}.
\]
\end{theorem}
\begin{proof}
Run the symbolic analysis through $\circuit_1 \;;\; \circuit_2$, and let
$\hat{\sigma}^{(1)}$ and $\hat{\sigma}^{(2)}$ be the symbolic states after
$\circuit_1$ and after $\circuit_1 \;;\; \circuit_2$, respectively.
Let $U_1$ be the Hadamard variables introduced while analyzing $\circuit_1$,
and let $U_2$ be those introduced while analyzing $\circuit_2$.
By \Cref{lem:symbolic-soundness}, applied to $\circuit_1$ and the pair
$(x_0,x)$, there is an assignment $h_1 \colon U_1 \to \{0,1\}$ such that
\[
  \alpha_{x_0,h_1}(\hat{\sigma}^{(1)}(\q{r})) = x_\q{r}
  \quad \text{for every qubit } \q{r}.
\]
Now continue the symbolic and concrete executions through $\circuit_2$, starting
from symbolic state $\hat{\sigma}^{(1)}$ and concrete state $x$.
The gate cases in the proof of \Cref{lem:symbolic-soundness} extend $h_1$ to an
assignment $h \colon U_1 \cup U_2 \to \{0,1\}$ such that
\[
  \alpha_{x_0,h}(\hat{\sigma}^{(2)}(\q{r})) = x'_\q{r}
  \quad \text{for every qubit } \q{r}.
\]
Since the variables in $U_2$ are fresh, they do not occur in
$\hat{\sigma}^{(1)}$, so this extension still satisfies
$\alpha_{x_0,h}(\hat{\sigma}^{(1)}(\q{q})) = x_\q{q}$.
Thus the same valuation $h$ makes
$\hat{\sigma}^{(1)}(\q{q})$ evaluate to $x_\q{q}$ and
$\hat{\sigma}^{(2)}(\q{q'})$ evaluate to $x'_\q{q'}$.
Since these values differ, the two parities are syntactically unequal.
The randomized analysis is the symbolic analysis evaluated under an independent
uniform random valuation, so \Cref{lem:parity-eq}(2) gives the stated
$2^{-k}$ bound.
\end{proof}

\section{A Linear-Time Phase-Folding Algorithm}
\label{sec:rand-algorithm}
We are now ready to present a linear-time algorithm for phase folding.
The algorithm uses the randomized abstraction from the previous section to detect which rotations can be combined.
The algorithm makes a single left-to-right pass over the input circuit $\circuit$ and, for a fixed bitstring width $k$, runs in expected $O(n + m)$ time,
where $n$ is the number of qubits and $m$ is the number of gates.
It outputs an optimized circuit $\circuit'$.

The algorithm, shown in \Cref{alg:fold}, maintains two structures:
\begin{itemize}
  \item the randomized abstract state $\sigma$, updated gate by gate using
    the transfer functions of \Cref{sec:randomized};
  \item a hashmap $M \colon \{0,1\}^k \rightharpoonup (\theta, \q{q})$ that,
    for each previously seen rotation $\Rz{\theta_1}\ \q{q_1}$, maps its abstract
    state $u = \sigma(\q{q_1})$ to the angle $\theta_1$ and qubit $\q{q_1}$ where
    a later matching rotation should be folded.
\end{itemize}
\Cref{thm:rand-soundness} makes $M$ the natural structure:
the merging condition $\sigma(\q{q_1}) = \sigma'(\q{q_2})$ reduces to
a hashmap lookup on a $k$-bit key, resolved in expected $O(1)$ time.

The main loop scans the circuit once, one gate at a time, from left to right.
On $\CX$, $\X$, and $\Had$ the algorithm updates $\sigma$ by the transfer
functions of \Cref{sec:randomized}.
On $\Rz{\theta}\ \q{q}$ it consults $M$ at key $u = \sigma(\q{q})$:
if an entry $(\theta', \q{q'})$ is present, the two rotations are merged---%
the earlier rotation is deleted from the output and the angle is
accumulated onto the later one, which remains at $\q{q}$;
if the combined angle is a multiple of $2\pi$, both rotations are eliminated.
Otherwise, $(\theta, \q{q})$ is inserted into $M$ at key $u$,
making the current rotation the new candidate for a future merge.

\begin{algorithm}[t]
\caption{Randomized phase folding.}
\label{alg:fold}
\begin{algorithmic}[1]
\Require circuit $\circuit = g_1 \,;\, \cdots \,;\, g_m$ on $n$ qubits; width $k$
\Ensure circuit $\circuit'$ equivalent to $\circuit$ with high probability
\State $\sigma(\q{q}) \gets \mathrm{Uniform}(\{0,1\}^k)$ for each qubit $\q{q}$
\State $M \gets \emptyset$ \Comment{hashmap $\{0,1\}^k \rightharpoonup (\theta, \q{q})$}
\State $\circuit' \gets \epsilon$
\For{each gate $g$ in $\circuit$}
  \If{$g = \Rz{\theta}\ \q{q}$}
    \State $u \gets \sigma(\q{q})$
    \If{$M(u) = (\theta', \q{q'})$ is defined}
      \State delete the emitted occurrence of $\Rz{\theta'}\ \q{q'}$ from $\circuit'$
      \State $\theta \gets \theta + \theta'$; \quad remove $u$ from $M$
      \If{$\theta \equiv 0 \pmod{2\pi}$} \textbf{continue} \EndIf
    \EndIf
    \State $M(u) \gets (\theta, \q{q})$
    \State append $\Rz{\theta}\ \q{q}$ to $\circuit'$
  \Else
    \State $\sigma \gets \rasem{g}(\sigma)$ \Comment{$\CX$, $\X$, or $\Had$ transfer}
    \State append $g$ to $\circuit'$
  \EndIf
\EndFor
\State \Return $\circuit'$
\end{algorithmic}
\end{algorithm}

\begin{example}
Consider the swap circuit from the overview (reproduced below) with $k = 3$,
initial bitstrings $\sigma(\q{q_1}) = 101$ and $\sigma(\q{q_2}) = 011$.
The table below traces the algorithm gate by gate, showing $\sigma$ and the hashmap $M$ after each step.

\begin{center}
\small
\begin{tabular}{lll}
\toprule
\textbf{Gate} & $\sigma$ & $M$ \\
\midrule
(init) & $\q{q_1} \mapsto 101,\ \q{q_2} \mapsto 011$ & $\emptyset$ \\
$\Rz{\theta_1}\ \q{q_1}$ & unchanged & $101 \mapsto (\theta_1, \q{q_1})$ \\
$\CX\ \q{q_1}\ \q{q_2}$ & $\q{q_2} \mapsto 110$ & unchanged \\
$\CX\ \q{q_2}\ \q{q_1}$ & $\q{q_1} \mapsto 011$ & unchanged \\
$\CX\ \q{q_1}\ \q{q_2}$ & $\q{q_2} \mapsto 101$ & unchanged \\
$\Rz{\theta_2}\ \q{q_2}$ & unchanged & $M(101)$ hit: merge into $\Rz{\theta_1+\theta_2}\ \q{q_2}$ \\
\bottomrule
\end{tabular}
\end{center}
At the second rotation, $\sigma(\q{q_2}) = 101$ matches the key stored by the first rotation,
so the algorithm folds $\theta_1$ into $\theta_2$ and deletes the earlier gate from the output.
\end{example}

\paragraph{Correctness}
Intuitively, the hashmap $M$ remembers, for every rotation emitted so far,
the bitstring that was sitting on its qubit at the moment it was emitted.
When a new rotation $\Rz{\theta}\ \q{q}$ arrives, the algorithm asks a single
question: does the current bitstring $\sigma(\q{q})$ match one that was
recorded earlier?
If it does, the two rotations have seen the same parity on every path
through the intervening circuit---exactly the condition of
\Cref{thm:rand-soundness}---and folding them is sound, up to the $2^{-k}$
chance of a spurious bitstring collision.
If $\sigma(\q{q})$ has since been disturbed---for instance by a $\Had$
on $\q{q}$, which draws a fresh uniform bitstring, or by a $\CX$ that mixes
another qubit into it---the new value almost certainly no longer matches the
stored key, the lookup misses, and no unsound merge is attempted.

Each individual pair of rotations suffers an unsound merge with probability at most $2^{-k}$, but the algorithm must succeed on \emph{all} pairs simultaneously.
To ensure that we have a safe accounting of the total error probability, we need to take a union bound over all pairs of rotations that the algorithm might attempt to merge.

A union bound over the at most $\binom{m}{2}$ rotation pairs yields the following correctness guarantee.

\begin{theorem}[Correctness]
\label{thm:alg-correctness}
Let $\circuit'$ be the output of \Cref{alg:fold}
on a circuit $\circuit$ with at most $m$ gates.
Then $\circuit' \equiv \circuit$ with probability at least $1 - \binom{m}{2} \cdot 2^{-k}$.
\end{theorem}

\begin{proof}
Enumerate the rotations of $\circuit$ in circuit order as
\[
  \Rz{\theta_1}\ \q{q_1},\; \ldots,\; \Rz{\theta_r}\ \q{q_r},
\]
with $r \leq m$, and let $\sigma^{(i)}$ denote the randomized abstract state
immediately before the $i$-th rotation is processed.

Call a pair $i < j$ \emph{unsafe} if, writing $\circuit_1$ for the prefix
before the $i$-th rotation and $\circuit_2$ for the subcircuit between the
$i$-th and $j$-th rotations, there are basis states $x_0,x,x'$ with
$\sem{\circuit_1}(x_0,x) \neq 0$,
$\sem{\circuit_2}(x,x') \neq 0$, and
$x_\q{q_i} \neq x'_\q{q_j}$.
For each unsafe pair, let $C_{i,j}$ be the event that their randomized
bitstrings nevertheless coincide,
$\sigma^{(i)}(\q{q_i}) = \sigma^{(j)}(\q{q_j})$.
By \Cref{thm:rand-soundness}, $\Pr[C_{i,j}] \leq 2^{-k}$.
There are at most
$\binom{m}{2}$ such pairs, so
\[
  \Pr\!\left[\bigcup_{i<j} C_{i,j}\right]
  \;\leq\; \sum_{i<j} \Pr[C_{i,j}]
  \;\leq\; \binom{m}{2} \cdot 2^{-k}.
\]

Condition on the complement of this event.
We show that every merge performed by the algorithm is semantics-preserving.
If the hashmap lookup at rotation $j$ hits an entry created by rotation $i$,
then $\sigma^{(i)}(\q{q_i}) = \sigma^{(j)}(\q{q_j})$.
Under the conditioned event, the pair $i,j$ is therefore not unsafe, so the
side condition of \Cref{thm:phase-merge} holds.
The merge deleting the earlier emitted rotation and adding its angle to the
current rotation is the phase-folding rewrite (\Cref{thm:phase-merge}).
All other gates are emitted unchanged.
Moreover, rotations do not change basis states, so earlier phase-folding
rewrites do not affect the reachability condition used by later applications of
\Cref{thm:phase-merge}.
\end{proof}

\paragraph{Complexity}
Our algorithm runs in time linear in the number of qubits $n$ and the number of gates $m$ in a given circuit. (Typically, $n$ is significantly smaller than $m$.) The following theorem captures that.
Intuitively, the algorithm makes a single forward pass through the circuit.
The only time it needs to \emph{backtrack} is when it needs to delete an already emitted rotation. This can be performed by keeping track of the address of the emitted rotation (a node in a linked list) and deleting it.

\begin{theorem}[Complexity]
\label{thm:complexity}
For a fixed bitstring width $k$, \Cref{alg:fold} runs in expected $O(n + m)$ time, where $n$ is the number of qubits and $m$ is the number of gates.
\end{theorem}
\begin{proof}
Initialization samples one bitstring per qubit, taking
$
  O(n) \text{ time.}
$
Each gate is processed in expected
$
  O(1) \text{ time:}
$
transfer functions touch at most two qubits and perform a constant number of
word-level XORs or a single uniform sample;
rotation handling performs a constant number of hashmap operations on $k$-bit keys,
which take expected $O(1)$ time.
The total is $O(n) + m \cdot O(1) = O(n + m)$.
\end{proof}

\paragraph{Choosing bitstring width $k$}
In practice, $k$ is fixed as a small
multiple of the machine word width. Our implementation uses $k=128$, giving an
error probability of $\sim 10^{-21}$ even for circuits with a \emph{billion}
gates---negligible in any practical setting. Thus, for all intents and purposes in our target regime,
bitstring operations are constant time and the circuit-size scaling is linear.

If, instead, one asks for a constant global error probability $\varepsilon$ as
the circuit size $m$ grows asymptotically, then $k$ must also grow with $m$:
by the correctness bound of \Cref{thm:alg-correctness}, it suffices to choose
$k \geq 2\log_2 m + \log_2(1/\varepsilon)$. In that regime, the total expected
time is $O((n+m)\log m)$ for constant $\varepsilon$.

\section{Implementation and Evaluation}
\label{sec:evaluation}

We describe our implementation and evaluation of the optimization technique presented above.

\paragraph{Implementation}
We implemented our algorithms in a tool called \ours in roughly 7{,}000 lines of Rust.
The tool is a command-line utility that ingests a subset of OpenQASM~2.0~\cite{cross2017openqasm} and
emits optimized OpenQASM~2.0. 
Before phase folding, \ours makes a pass over the circuit to remove
redundant sequences of gates (e.g., $XX$, $\Had\Had$).

\paragraph{Research questions}
We aim to answer the following research questions:
\begin{enumerate}
  \item How does \ours compare to existing optimizers in
    runtime and $T$-count reduction?
  \item How does \ours scale to very large circuits, with millions to
    hundreds of millions of gates?
  \item How does \ours perform with a symbolic abstract state
    that tracks exact parity formulas instead of random bitstrings?
\end{enumerate}

\paragraph{Benchmarks}
We evaluate \ours on two benchmark suites spanning five orders of magnitude in
circuit size.

The first is the standard
suite~\cite{amy2025linear} from the Feynman optimizer, which has
become the \emph{de facto} benchmark set for $T$-count optimization. It
collects arithmetic and oracle circuits drawn from the quantum-algorithms
literature: Galois-field multipliers ($\mathrm{GF}(2^k)$ for $k$ ranging from
$4$ to $256$), modular and carry-lookahead adders, hidden-weighted-bit
functions, Barenco Toffoli decompositions, Hamming-weight circuits, and
several reversible-logic benchmarks (Grover oracles, cycle and mod-$5$
circuits, etc.). Sizes in this suite range from tens of gates up to circuits
with hundreds of thousands of gates after \cliffordt decomposition.

The second is the Cobble benchmark suite~\cite{yuan2024cobble}, which contains larger and more
structurally varied circuits drawn from two application domains:
(i)~Hamiltonian simulation circuits and (ii)~quantum linear-algebra kernels, including Chebyshev
polynomial evaluation, Laplacian filtering, matrix inversion, ordinary
least-squares ridge regression, and spectral thresholding. 

Across both suites, we use the Clifford+$T$ forms of the circuits: any
$R_z(\theta)$ rotations with non-$\pi/4$-multiple angles are decomposed using
\texttt{gridsynth}~\cite{ross2016optimal} at $\epsilon = 10^{-10}$. 
This ensures a
uniform comparison across tools by running each tool on Clifford+$T$ input.
We run all tools on an Apple M3 CPU with 24 GB of RAM, and we set a one-hour timeout for each benchmark.
Throughout, we used $k=128$ for the randomized analysis.

\paragraph{Correctness validation}
We validated the correctness of our implementation at two scales.
For small circuits (up to 6 qubits), we generated hundreds of thousands of random circuits and verified that the output of \ours is equivalent to the input by comparing their matrix representations directly.
For larger circuits, we used Feynman's path-sum-based equivalence verifier~\cite{amy2018verification} to check equivalence between input and output on a majority of the feasible circuits in Feynman's suite.
All validation runs succeeded: every checked output of \ours was equivalent to its input.

\subsection{RQ1: Comparison to existing tools}
\label{sec:rq1}

\begin{takeaway}
On both benchmark suites, \ours closely matches the $T$-count reductions of VOQC,
QuiZX, PyZX, and Feynman to within a few percent---and typically matches them
exactly on the largest circuits---while running up to four orders of magnitude
faster. On the largest circuits, \ours completes in milliseconds while other tools time out even with a one-hour budget.
\end{takeaway}

We compare \ours against five optimizers spanning the main
families of $T$-count reduction techniques.
VOQC~\cite{hietala2021verified} is a verified optimizer that applies the
Nam et al.~\cite{nam2018automated} rewriting and pass-scheduling protocol.
\pyzx~\cite{kissinger2020pyzx} and its Rust reimplementation QuiZX~\cite{kissinger2021quizx} optimize
via ZX-calculus rewriting and graph extraction.
Feynman~\cite{amy2014polynomial,amy2025linear}, whose phase-polynomial affine
analysis is the closest in spirit to our approach, performs phase folding
through an exact symbolic representation of the parities at each rotation.\footnote{We ran Feynman with the \texttt{-apf} flag to run the fastest (affine) phase-folding pass, which is the most comparable to \ours.}
We also evaluated FastTODD~\cite{vandaele2024lower}, but found it to be significantly slower than all of the above tools on our benchmarks; we therefore exclude it from the comparison.

\paragraph{Feynman suite}
\Cref{fig:feynman-results} reports $T$-count reduction and runtime on the Feynman
suite. 
We focus on the benchmarks with more than $2{,}000$ \cliffordt gates: below that size all tools finish in milliseconds, so runtime comparisons are dominated by noise.
All tools produce very similar $T$-count reduction profiles across
the suite: \ours matches VOQC, QuiZX, PyZX, and Feynman exactly on the large
Galois-field multipliers ($\mathrm{GF}(2^k)$ for $k=8,\dots,64$), and \ours is within
a few percent of the others on the remaining circuits.

On runtime, \ours is multiple orders of magnitude faster than the alternatives across
the entire suite. On \texttt{gf2\string^64\_mult}---a $70{,}075$-gate
circuit and the largest on which VOQC, QuiZX, and Feynman all still
complete---\ours finishes in $17.9$\,ms, compared to $5.3$\,minutes for
VOQC, $17.7$\,minutes for QuiZX, $30.7$\,minutes for Feynman, and a
one-hour timeout for PyZX. The gap typically reaches four orders of magnitude on large benchmarks and widens
with circuit size. Indeed, on the largest circuits, all tools time out after one hour while \ours completes the optimization in milliseconds.

\begin{figure}[t]
  \centering
  \begin{subfigure}{\linewidth}
    \centering
    \includegraphics[width=0.99\linewidth]{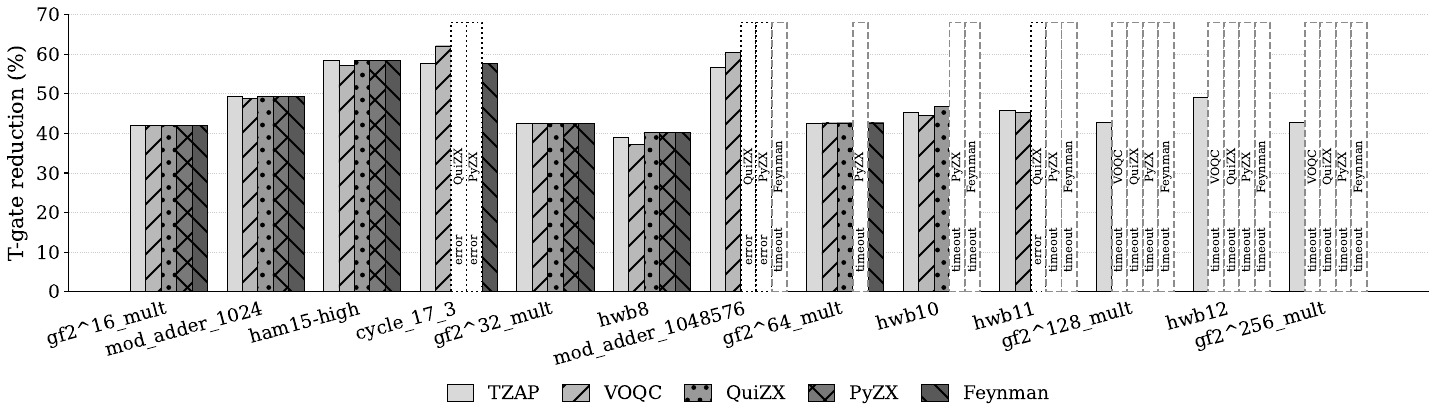}
    \caption{$T$-count after optimization.}
    \label{fig:feynman-t}
  \end{subfigure}\\[1ex]
  \begin{subfigure}{\linewidth}
    \centering
    \includegraphics[width=0.99\linewidth]{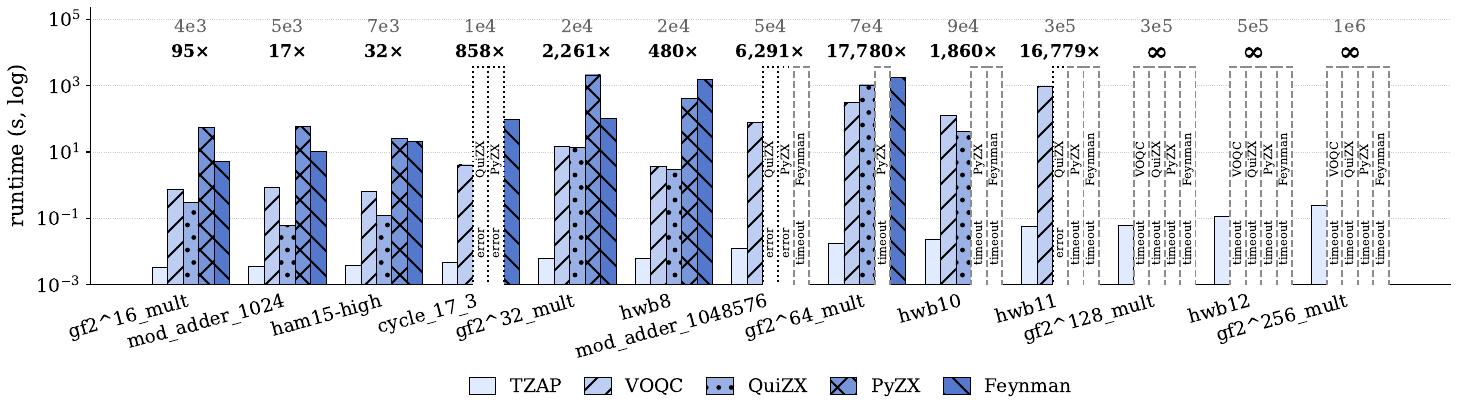}
    \caption{Runtime (\textbf{log} scale). 
    On top of each benchmark, we show the number of gates (gray) and the speedup of \ours over the fastest other tool (black).}
    \label{fig:feynman-runtime}
  \end{subfigure}
  \caption{Results on the Feynman benchmark suite. 
  The circuits are sorted by total gate count.
  Missing bars indicate
    that the tool crashed or timed out ($> 1$\,hour).}
  \label{fig:feynman-results}
\end{figure}

\begin{figure}[t]
  \centering
  \begin{subfigure}{1\linewidth}
    \centering
    \includegraphics[width=.65\linewidth]{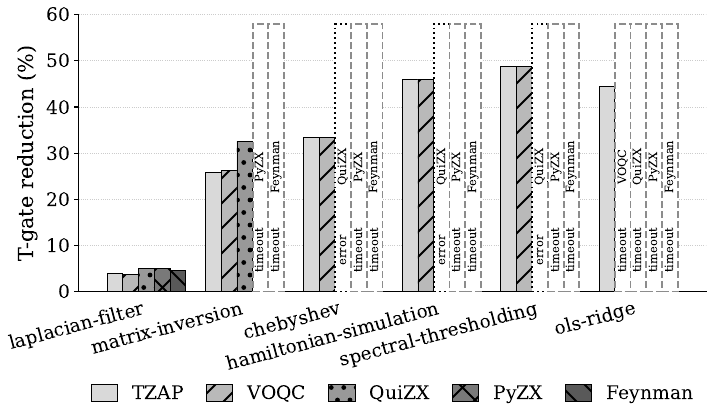}
    \caption{$T$-count after optimization.}
    \label{fig:cobble-t}
  \end{subfigure}\\[1ex]
  \begin{subfigure}{1\linewidth}
    \centering
    \includegraphics[width=.65\linewidth]{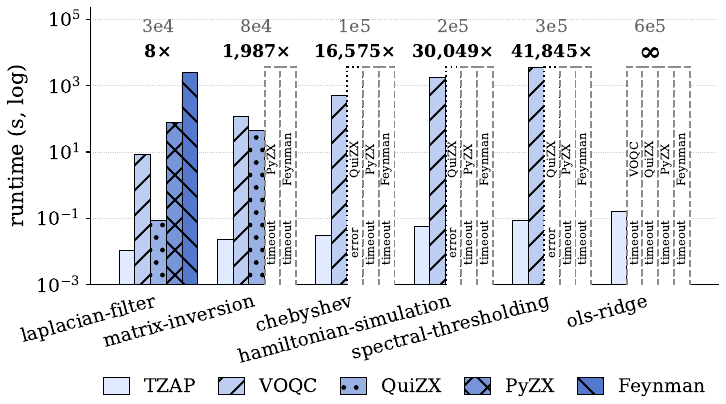}
    \caption{Runtime (\textbf{log} scale).}
    \label{fig:cobble-runtime}
  \end{subfigure}
  \caption{$T$-count reduction and runtime on the Cobble benchmark suite.
   Circuits are sorted by total gate count.
    Missing bars indicate that the tool crashed or timed out
    ($> 1$\,hour) on that circuit. For (b), on top of each benchmark, we show the number of gates (gray) and the speedup of \ours over the fastest other tool (black).}
  \label{fig:cobble-results}
\end{figure}

\paragraph{Cobble suite}
\Cref{fig:cobble-results} reports the results on the Cobble suite.
On all six circuits, \ours completes in well under a second---between $11$\,ms and $164$\,ms.

The other tools fare considerably worse.
VOQC is the only competitor to complete on most of the suite, taking from $8.5$\,s on \texttt{laplacian-filter} up to roughly an hour on \texttt{spectral-thresholding}, and timing out on \texttt{ols-ridge} (with $255{,}692$ input $T$ gates).
QuiZX, PyZX, and Feynman fail on the majority of Cobble benchmarks: QuiZX's ZX-extraction pass crashes on three of the six circuits and times out on \texttt{ols-ridge}, and PyZX and Feynman hit the $1$-hour timeout on all but \texttt{laplacian-filter}.

On \texttt{laplacian-filter}, the one circuit where all five tools run to completion, \ours is $775\times$ faster than VOQC, $7{,}020\times$ faster than PyZX, and $233{,}000\times$ faster than Feynman.

$T$-count reduction is competitive across the suite. On \texttt{chebyshev},
\texttt{hamiltonian-simulation}, and \texttt{spectral-thresholding}, \ours
lands within $0.3\%$ of VOQC's $T$-count. On \texttt{laplacian-filter},
QuiZX achieves $1$ percentage point more $T$-count reduction than \ours
($5.0\%$ vs.\ $3.9\%$), and on \texttt{matrix-inversion} QuiZX achieves
$6.7$ percentage points more ($32.6\%$ vs.\ $25.9\%$) at the cost of a
$1{,}990\times$-longer runtime. On
\texttt{ols-ridge}, where every other tool times out, only \ours produces a
result, reducing the $T$-count by $44\%$ in $164$\,ms.

\paragraph{Isolation experiments}
Tools like VOQC and Feynman bundle phase folding with additional
simplification passes, for example, gate cancellation passes, propagation of NOTs to the end of the circuit, Hadamard cancellation, and others. VOQC follows the Nam et al.~\cite{nam2018automated} pass-scheduling
protocol, and Feynman---even when restricted to its affine phase-folding mode
(\texttt{-apf})---iteratively applies gate cancellation and phase folding
a number of times. To isolate the cost of phase folding itself and put the
runtime comparison on equal footing, we modified the source code of VOQC and Feynman to perform
a single phase-folding pass. For Feynman, the pass is \texttt{PhaseFold} and for VOQC it is \texttt{merge-rotations}.

\Cref{fig:cobble-pf} shows the same isolation experiment on the Cobble suite
(the Feynman suite is reported in \Cref{app:feynman-pf}),
comparing \ours (\ourspf) against the phase-folding-only variants of VOQC
(VOQC-PF) and Feynman (Feynman-PF).
Three trends stand out.

First, restricting each tool to a single phase-folding pass naturally lowers
its $T$-count reduction relative to its full pipeline---the additional
cancellation and rewriting passes expose further folding opportunities---but
the three phase-folding passes then agree closely with one another.
On \texttt{chebyshev}, for instance, the full pipelines of \ours and VOQC
reduce the $T$-count by $33.5\%$ and $33.4\%$, whereas all three
phase-folding passes land at $14.4\%$--$14.5\%$; on \texttt{matrix-inversion}
they agree to within $0.1$ percentage point ($11.7\%$--$11.8\%$).

Second, isolating phase folding speeds up the two competitors by very
different amounts. VOQC sees a smaller speedup---between $0.9\times$ and
$2.3\times$ (e.g., $523$\,s down to $224$\,s on \texttt{chebyshev})---confirming
that phase folding already dominates its runtime. Feynman, by contrast, speeds
up substantially: on \texttt{laplacian-filter} it drops from $2{,}557$\,s to
$21$\,s ($124\times$), and Feynman-PF completes three of the six circuits
within the one-hour budget, where full Feynman completes only one.

Third, despite these speedups the orders-of-magnitude gap to \ours remains.
\ourspf finishes every Cobble circuit in at most $90$\,ms, whereas on
\texttt{chebyshev} VOQC-PF still takes $224$\,s and Feynman-PF $819$\,s---about
$11{,}000\times$ and $41{,}000\times$ slower---and both still time out or crash
on the largest circuits that \ours handles in milliseconds.
The gap also widens with circuit size: the competing phase-folding passes
appear, empirically, not to scale linearly with growing circuit size, whereas
\ours does (\Cref{sec:rq2}).

In short, the runtime gap to \ours appears to be primarily driven by the representation and the per-gate cost of the analysis, not by extra passes layered on top. 
This is perhaps expected, as VOQC looks for pairs of rotations, creating a quadratic search space, and both Feynman and VOQC maintain a symbolic representation of the parities at each rotation, which is more expensive to manipulate than our randomized bitstring abstraction.
(We study this in more detail in \Cref{sec:rq3}.)

\begin{figure}[t]
  \centering
  \begin{subfigure}{.49\linewidth}
    \centering
    \includegraphics[width=\linewidth]{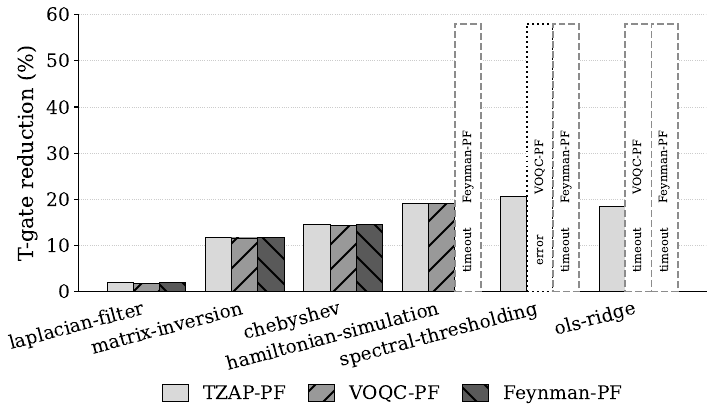}
    \caption{$T$-count after optimization.}
    \label{fig:cobble-pf-t}
  \end{subfigure}
  \begin{subfigure}{.49\linewidth}
    \centering
    \includegraphics[width=\linewidth]{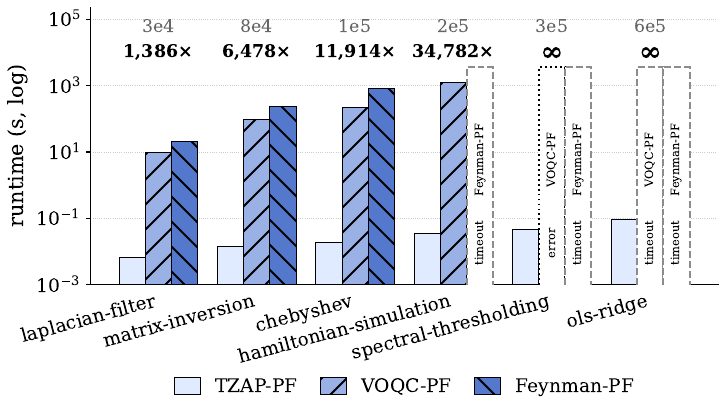}
    \caption{Runtime (\textbf{log} scale).}
    \label{fig:cobble-pf-runtime}
  \end{subfigure}
  \caption{Phase-folding isolation on the Cobble benchmark suite. The circuits are sorted by total gate count.
    Missing bars indicate that the tool crashed or timed out
    ($> 1$\,hour) on that circuit. For (b), on top of each benchmark, we show the number of gates (gray) and the speedup of \ours over the fastest other tool (black).}
  \label{fig:cobble-pf}
\end{figure}

\subsection{RQ2: Scalability}
\label{sec:rq2}

\begin{takeaway}
\ours scales linearly over four orders of magnitude in circuit size,
matching the $O(n+m)$ analysis of \Cref{sec:rand-algorithm}. It optimizes a
$148$\,million-gate Cobble instance in under a minute and a
$\sim\!500$\,million-gate Tower instance in about two minutes, while still
delivering an average $T$-count reduction of $> 50\%$ across the scaling
experiments.
\end{takeaway}

To stress-test scalability beyond the sizes found in the standard benchmark
suites, we took two of the Cobble circuits---Hamiltonian simulation and
matrix inversion---and used Cobble to generate increasingly large instances of each by varying the parameters, with
the largest reaching $148$\,million gates.
We then ran \ours on each instance and measured end-to-end runtime.
\Cref{fig:scalability} plots runtime against input $T$-count on a log--log
scale.

\begin{figure}[t]
  \centering
  \includegraphics[width=0.33\linewidth]{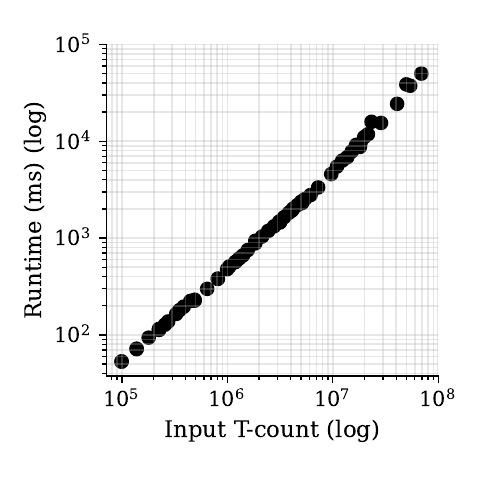}
  ~~~
  \includegraphics[width=0.33\linewidth]{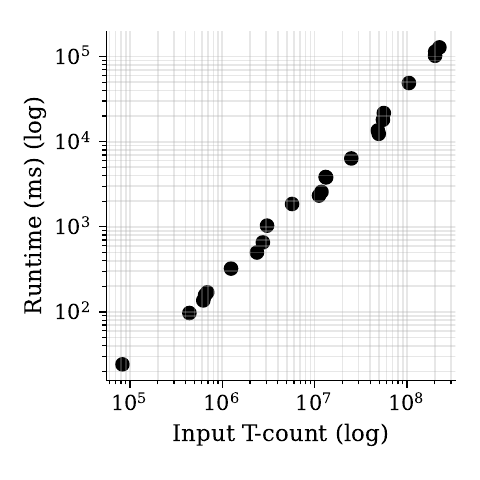}
  \caption{Runtime of \ours on increasingly large circuits (log--log scale). \emph{Left:} Hamiltonian simulation \& matrix inversion. \emph{Right:} Tower benchmarks (data structures).}
  \label{fig:scalability}
\end{figure}

The runtime scales linearly in the size of the circuit, as predicted by the
$O(n + m)$ analysis of \Cref{sec:rand-algorithm}: on the log--log plot the data
follow a straight line of slope one across four orders of magnitude. Even on
the largest $148$\,million-gate instance, \ours finishes in under a minute.
We note that these are non-trivial circuits: \ours is doing substantial work on each circuit,
reducing the $T$-count by an average of $57\%$ across the scaling
experiments.

To further ensure that the scalability results hold,
we repeat the same exercise on the Tower benchmarks~\cite{yuan2022tower}, a very different suite of quantum data-structure manipulation circuits. Again, we used Tower to generate large instances reaching
roughly $500$\,million gates (\Cref{fig:scalability}, right). The same linear
trend holds, with \ours completing the largest instance in about two
minutes, and again delivering significant $T$-count reductions.

\subsection{RQ3: Symbolic abstract state}
\label{sec:rq3}

\begin{takeaway}
Replacing the randomized bitstring abstraction with an exact symbolic
parity abstraction---following the same algorithmic skeleton of
\Cref{alg:fold}---makes the analysis between $6\times$ and $25\times$
slower, while producing the same $T$-count reductions.
On certain circuits, the slowdown grows with circuit size, resulting in \emph{many orders of magnitude} of difference on large circuits.
\end{takeaway}

A natural question to ask is how much of \ours's runtime advantage comes from
the randomized abstraction itself, as opposed to \Cref{alg:fold}, which
performs a single linear scan through the circuit.
To isolate this, we implemented the symbolic variant of the analysis---which we first sketched out
in \Cref{sec:overview} and formalized in \Cref{sec:soundness} for reasoning about soundness---within \ours: the abstract state still maps
each qubit to a parity, but parities are represented as symbolic
formulas, rather than as random
bitstrings. All other components of \Cref{alg:fold}---the
per-gate transfer functions, the table of seen parities, and the merging
rule---are unchanged.
We note that the symbolic analysis is similar to that implemented in Feynman and VOQC.

In principle this variant is no longer a linear-time algorithm: parity
formulas can grow with the number of Hadamard-introduced variables, so
hashing and equality checks are no longer constant time. Hadamard gates
do bound this growth---each Hadamard resets the parity of the qubit it
touches to a fresh variable---but the $\CX$ network between Hadamards mixes these fresh
variables together, and the \emph{weight} of the resulting parity
formulas (the number of variables they XOR together) grows with the
density of Hadamards in the circuit. Every symbolic operation---XOR, hashing, equality---then pays a cost proportional to this
weight, on every gate. We made a careful effort to implement the
symbolic variant well, using compact normal forms.

\begin{wrapfigure}{r}{0.38\linewidth}
  \centering
  \vspace{-1ex}
  \includegraphics[width=\linewidth]{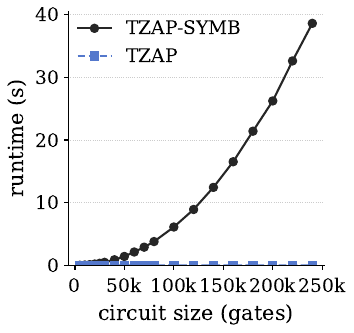}
  \caption{Symbolic variant (\ourssymb) versus \ours on the
    \texttt{hwb}-like family of \Cref{app:hwb-like}.}
  \label{fig:hwb-like-scaling}
  \vspace{-2ex}
\end{wrapfigure}

Across our benchmarks, the symbolic variant is $6\times$ to $25\times$
slower than \ours, with the $T$-count reductions matching \ours exactly,
as expected: the two analyses fold the same set of rotations, with high
probability. 
We notice that the slowdown is not uniform across circuits, but rather depends on the structure of the circuit and the density of Hadamards.
Circuits with regular, low-weight parity structure, such as the
arithmetic adders and Galois-field multipliers, sit at the low end
(around $6\times$), because their parity formulas stay small. The
largest slowdowns, up to $25\times$, come from the \texttt{hwb}
(hidden-weighted-bit) circuits, which place thousands of Hadamards on
only a dozen or so qubits: the dense interleaving of Hadamards and
$\CX$ gates produces high-weight parity formulas, exactly the case the
exact representation handles worst. The randomized abstraction, by
contrast, represents each parity as a single fixed-width bitstring
regardless of its weight, so its per-gate cost is unaffected.

To isolate this effect at scale, we constructed a parameterized family
of \texttt{hwb}-like circuits that distills the structure responsible
for the slowdown: an accumulator qubit collects one fresh
Hadamard-introduced variable per stage, so the largest parity formula
grows linearly with circuit size (the construction is described in
\Cref{app:hwb-like}).
\Cref{fig:hwb-like-scaling} plots both variants on this family with
linear axes: the symbolic variant's (\ourssymb) runtime grows quadratically with
circuit size---each gate pays a cost proportional to the current parity
weight---reaching $39$ seconds at $240$K gates, while \ours remains
linear and completes the same circuit in under $3$ milliseconds, a speedup of $13{,}000\times$!

\section{Related Work}
Quantum circuit optimization has attracted many different approaches, ranging from algebraic rewriting and ZX-calculus graph rewriting to numerical synthesis and learning-based superoptimization.
Our work fits the \emph{static analysis} tradition: a single forward pass over the circuit, propagating an abstract state through each gate and using the result to identify and fold redundant rotations.
This is analogous to classical data-flow analyses and constant-propagation passes, and it is closest in spirit to the linear analyses of \citet{amy2025linear}, VOQC~\cite{hietala2021verified}, \citet{nam2018automated}, and QUESO~\cite{xu2022synthesizing}, which also maintain an explicit symbolic representation of the parities at each rotation site.
We organize the related work into five groups: general circuit optimizers, $T$-gate optimizers (the focus of this paper), abstract interpretation for quantum programs, probabilistic static analysis, and symbolic/concolic execution.

\paragraph{Quantum circuit optimization}
VOQC~\cite{hietala2021verified} is a formally verified optimizer; it implements the Nam et al.~\cite{nam2018automated} pass-scheduling protocol, iterating gate cancellation and phase folding.
BQSKit~\cite{younis2022bqskit} and tket~\cite{sivarajah2020tket} are widely used end-to-end frameworks: BQSKit uses parameterized circuit instantiation as a core primitive for gate-count optimization and gate-set decomposition, and tket provides retargetable optimization and routing passes for NISQ devices.

A second line of work synthesizes the rewrite rules themselves.
QUESO~\cite{xu2022synthesizing} and Quartz~\cite{xu2022quartz} automatically derive provably correct rewrite rules for a target gate set.
GUOQ~\cite{xu2025optimizing} combines fast rewrite-rule-based optimization with slow search-based unitary resynthesis as in BQSKit.

A third line operates in the ZX-calculus~\cite{coecke2011interacting}. \citet{duncan2020graph} showed how to simplify circuits by converting to ZX-diagrams and applying various rewrite rules; \pyzx~\cite{kissinger2020pyzx} and its Rust reimplementation QuiZX~\cite{kissinger2021quizx} implement such ideas and other optimizations.

\paragraph{$T$-gate optimization}
In fault-tolerant quantum computing---e.g., using the surface code~\cite{fowler2012surface}---$T$ gates cannot be implemented transversally~\cite{bravyi2005universal} and instead require costly magic-state-distillation protocols~\cite{bravyi2005universal,bravyi2012magic,gidney2024cultivation}, making each logical $T$ gate significantly more expensive than a Clifford gate.
$T$-count minimization is therefore a central concern for large-scale quantum computation, further reinforced by the fact that classical simulation of \cliffordt circuits runs in time exponential in the $T$-count~\cite{bravyi2016simulation}.
Van de Wetering and Amy~\cite{vandewetering2024hard} prove that $T$-count optimization (and related metrics) is NP-hard in general, motivating the heuristic and approximate approaches.

Selinger~\cite{selinger2013quantum} showed that any \cliffordt circuit can be rewritten to have $T$-depth one, establishing foundational bounds.
Amy et al.~\cite{amy2014polynomial} gave a polynomial-time algorithm for $T$-count and $T$-depth optimization via matroid partitioning.
Amy and Mosca~\cite{amy2019tcount} connect minimum $T$-count to minimum-weight Reed--Muller codewords, an idea that has been exploited in subsequent work by  \citet{heyfron2019efficient} and \citet{vandaele2024lower} in the TODD family of optimizers.
Amy and Lunderville~\cite{amy2025linear} recently recast phase folding as a relational static analysis, like our formulation here, and extended it to programs with arbitrary classical control flow.

\citet{yuan2024tcomplexity} work one layer up the stack: their Spire compiler restructures quantum control flow \emph{before} circuit-level compilation to avoid the  $T$-count blowup that fault-tolerant control constructs can otherwise introduce, complementary to circuit-level phase folding.

\paragraph{Abstract interpretation for quantum programs}
Abstract interpretation~\cite{cousot1977abstract} has been applied to quantum programs to prove properties without full simulation.
\citet{yu2021quantum} develop an abstract domain over density matrices, enabling verification of assertions for quantum circuits.
\citet{bichsel2023abstraqt} introduce Abstraqt, which compresses stabilizer simulation into a single abstract summand to overapproximate circuit behavior.
Both target \emph{verification}; we share the abstract-interpretation philosophy but use the abstraction to drive an \emph{optimization}, accepting a small probability of unsoundness in exchange for a fast, linear-time algorithm.

\paragraph{Probabilistic static analysis}
Our work can be seen as a straight-line, quantum-circuit instantiation of Gulwani and Necula's \emph{random interpretation}~\cite{gulwani2003discovering,gulwani2004gvn,gulwani2005precise,gulwani2005thesis}, a randomized form of abstract interpretation~\cite{cousot1977abstract} that trades absolute soundness for a vanishingly small error probability in exchange for simpler and faster algorithms.
We share the core philosophy: track an affine abstraction by random instances, and use a hash-style equality check in place of symbolic reasoning.

\paragraph{Symbolic execution}
Symbolic execution~\cite{king1976symbolic} runs a program on symbolic inputs and maintains a formula describing all reachable states.
The analysis underlying tools employing phase folding can be viewed as symbolic execution over quantum circuits: each qubit is assigned a symbolic parity expression over the initial inputs and fresh Hadamard-path variables, updated by transfer functions for each gate, and two rotations can be folded precisely when their parity expressions are provably equal.
Concolic execution~\cite{sen2005cute} blends concrete and symbolic execution in tandem: a program runs on a concrete input while simultaneously accumulating a symbolic formula for the current path, which is then solved to generate new concrete inputs that explore different paths.
Our approach can be seen as a further simplification: we run only a single concrete execution on a randomly sampled bitstring, dispensing with the symbolic component entirely.

\section{Discussion and Future Work}
We presented an algorithm for optimizing quantum circuits, focused on reducing the number of $T$ gates.
The key technical insight underlying our approach is a randomized static analysis that tracks per-qubit parities via random bitstrings, enabling a linear-time phase-folding algorithm that is correct with high probability.
Our implementation, \ours, delivers $T$-count reductions comparable to state-of-the-art optimizers while running up to multiple orders of magnitude faster, completing in milliseconds optimizations that time out for other tools, and scaling linearly to circuits with hundreds of millions of gates.

Our approach raises a number of questions for future consideration.
First, our abstraction tracks per-qubit parities (XORs and negations) and treats Hadamards in a non-deterministic way, potentially losing optimization opportunities. Thus, a natural question to ask is whether we can enrich the analysis while maintaining its linear-time nature. For example, one approach may be to track superposition up to a fixed limit, and then introduce non-determinism to limit the explosion. This is an approach that was used by~\citet{bichsel2023abstraqt} in the context of quantum program verification, and it would be interesting to see if it can be adapted to the optimization setting.

Another potential extension is handling classical--quantum programs, like repeat-until-success loops, which are common in quantum algorithms and can be optimized by folding rotations across loop iterations. This would require extending our abstract domain to handle loops and control flow, perhaps by combining the techniques of Gulwani and Necula~\cite{gulwani2003discovering,gulwani2005thesis} with Amy and Lunderville's relational analyses for quantum programs~\cite{amy2025linear}.

Finally, it would be interesting to explore how our fast analysis can be used as a lightweight static analysis---for example, to prove assertions about parity properties of qubits at various points in the circuit, or to drive other optimizations beyond phase folding.

\begin{acks}
The authors thank Tom Reps, Charles Yuan, Abtin Molavi, Amanda Xu, Stavan Jain, Keshav Sharma, and Tianyi Hao for helpful discussions and feedback. The work is supported by NSF Award \#2212232.
\end{acks}

\bibliographystyle{ACM-Reference-Format}
\bibliography{refs}

\newpage
\appendix

\section{Soundness Details}
\label{app:soundness-details}

\begin{proof}[Proof of \Cref{lem:symbolic-soundness}]
By induction on the number of gates in $\circuit$.

\emph{Base case.}
$\circuit$ is the empty circuit.
Then $\sem{\circuit}$ is the identity relation, so $\sem{\circuit}(x, x') \neq 0$ implies $x' = x$; moreover $\hat{\sigma} = \hat{\sigma}_0$ and $U = \emptyset$, so $h$ is the empty assignment.
Indeed,
$\alpha_{x,h}(\hat{\sigma}(\q{q})) = \alpha_{x,h}(v_\q{q}) = x_\q{q} = x'_\q{q}$.

\emph{Inductive step.}
Let $\circuit = \circuit_0 \,;\, g$ for a single gate $g$, and let
$\hat{\sigma}_1 = \sasem{\circuit_0}(\hat{\sigma}_0)$, so that $\hat{\sigma} = \sasem{g}(\hat{\sigma}_1)$.
Let $U_0$ be the Hadamard variables of $\circuit_0$.

\emph{Inductive hypothesis:}
for all basis states $x, x'$ with $\sem{\circuit_0}(x, x') \neq 0$, there exists an assignment $h_0 \colon U_0 \to \{0,1\}$ such that $\alpha_{x,h_0}(\hat{\sigma}_1(\q{q})) = x'_\q{q}$ for every qubit $\q{q}$.

Suppose $\sem{\circuit}(x, x'') \neq 0$.
By the composition rule,
\[
  \sem{\circuit}(x, x'') \;=\; \sum_{x'} \sem{\circuit_0}(x, x') \cdot \sem{g}(x', x''),
\]
so at least one summand is nonzero: there is a basis state $x'$ with $\sem{\circuit_0}(x, x') \neq 0$ and $\sem{g}(x', x'') \neq 0$.
Let $h_0$ be the assignment given by the inductive hypothesis for the pair $(x, x')$, and write $\alpha_0 = \alpha_{x,h_0}$.
We now case on $g$, using the gate semantics of \Cref{fig:gate-semantics}.

\begin{itemize}
\item $g = \CX\ \q{c}\ \q{t}$:
then $x'' = x'[\q{t} \mapsto x'_\q{t} \oplus x'_\q{c}]$, and the analysis updates only $\q{t}$, setting $\hat{\sigma}(\q{t}) = \hat{\sigma}_1(\q{t}) \oplus \hat{\sigma}_1(\q{c})$.
No new variable is introduced, so take $h = h_0$:
\[
  \alpha_0(\hat{\sigma}(\q{t})) = \alpha_0(\hat{\sigma}_1(\q{t})) \oplus \alpha_0(\hat{\sigma}_1(\q{c})) = x'_\q{t} \oplus x'_\q{c} = x''_\q{t},
\]
and for every $\q{r} \neq \q{t}$, $\alpha_0(\hat{\sigma}(\q{r})) = \alpha_0(\hat{\sigma}_1(\q{r})) = x'_\q{r} = x''_\q{r}$.

\item $g = \X\ \q{q}$:
then $x'' = x'[\q{q} \mapsto x'_\q{q} \oplus 1]$ and $\hat{\sigma}(\q{q}) = \hat{\sigma}_1(\q{q}) \oplus 1$.
Take $h = h_0$:
$\alpha_0(\hat{\sigma}(\q{q})) = \alpha_0(\hat{\sigma}_1(\q{q})) \oplus 1 = x'_\q{q} \oplus 1 = x''_\q{q}$,
and all other qubits are untouched.

\item $g = \Rz{\theta}\ \q{q}$:
then $x'' = x'$ and $\hat{\sigma} = \hat{\sigma}_1$, so $h = h_0$ works.

\item $g = \Had\ \q{q}$:
then $x''$ agrees with $x'$ on every qubit except possibly $\q{q}$, and the analysis sets $\hat{\sigma}(\q{q}) = u$ for a fresh Hadamard variable $u$, so $U = U_0 \cup \{u\}$.
Extend the assignment to $u$: $h = h_0[u \mapsto x''_\q{q}]$.
Because $u$ is fresh, it occurs in no parity of $\hat{\sigma}_1$, so $\alpha_{x,h}$ agrees with $\alpha_0$ on every such parity.
Hence for every $\q{r} \neq \q{q}$,
$\alpha_{x,h}(\hat{\sigma}(\q{r})) = \alpha_0(\hat{\sigma}_1(\q{r})) = x'_\q{r} = x''_\q{r}$,
and $\alpha_{x,h}(\hat{\sigma}(\q{q})) = h(u) = x''_\q{q}$ by construction.
\end{itemize}

In every case we have produced an assignment of the Hadamard variables witnessing the pair $(x, x'')$, completing the induction.
\end{proof}

Note the soundness reading of the lemma:
since $\hat{\sigma}_0(\q{q}) = v_\q{q}$, we have $\alpha_{x,h}(\hat{\sigma}_0(\q{q})) = x_\q{q}$, so the initial and final parities are evaluated under one shared valuation.
So if two parities are equal---say $\hat{\sigma}_0(\q{q}) = \hat{\sigma}(\q{q'})$---then $x_\q{q} = x'_\q{q'}$ for \emph{every} pair with $\sem{\circuit}(x, x') \neq 0$: every equality the symbolic analysis claims is genuine.
Contrapositively, if $x_\q{q} \neq x'_\q{q'}$ for some realizable pair, the parities $\hat{\sigma}_0(\q{q})$ and $\hat{\sigma}(\q{q'})$ differ as parities.

\section{Phase-Folding Isolation on the Feynman Suite}
\label{app:feynman-pf}

\Cref{fig:feynman-pf} reports the phase-folding isolation experiment of
\Cref{sec:rq1} (``Isolation experiments'') on the Feynman suite, mirroring
\Cref{fig:cobble-pf} for the Cobble suite: \ours (\ourspf) against the
phase-folding-only variants of VOQC (VOQC-PF) and Feynman (Feynman-PF). As on
Cobble, the three passes agree closely on $T$-count where they complete, while
\ours runs orders of magnitude faster and completes every circuit.

\begin{figure}[t]
  \centering
  \begin{subfigure}{\linewidth}
    \centering
    \includegraphics[width=0.99\linewidth]{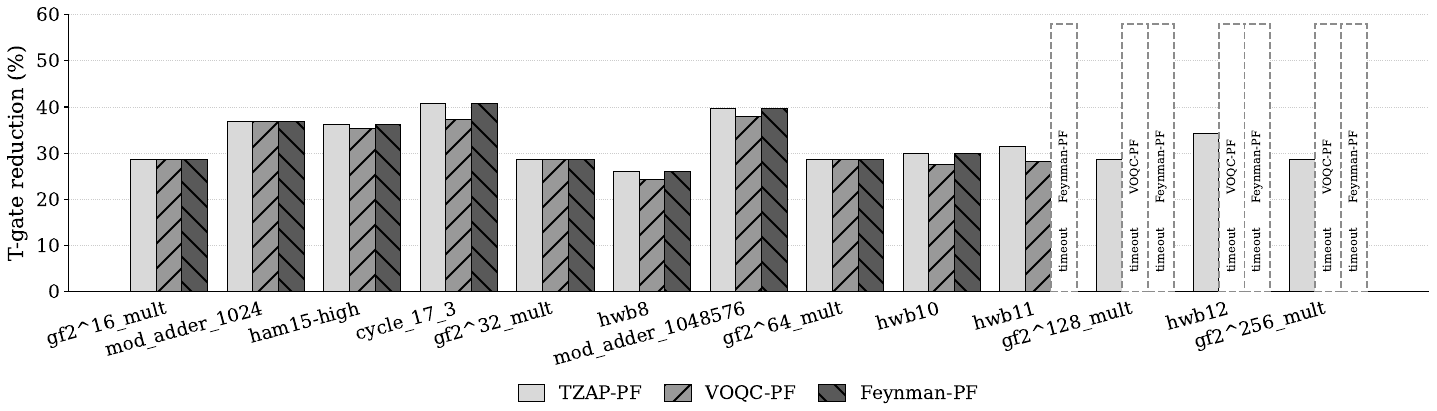}
    \caption{$T$-count after optimization.}
    \label{fig:feynman-pf-t}
  \end{subfigure}\\[1ex]
  \begin{subfigure}{\linewidth}
    \centering
    \includegraphics[width=0.99\linewidth]{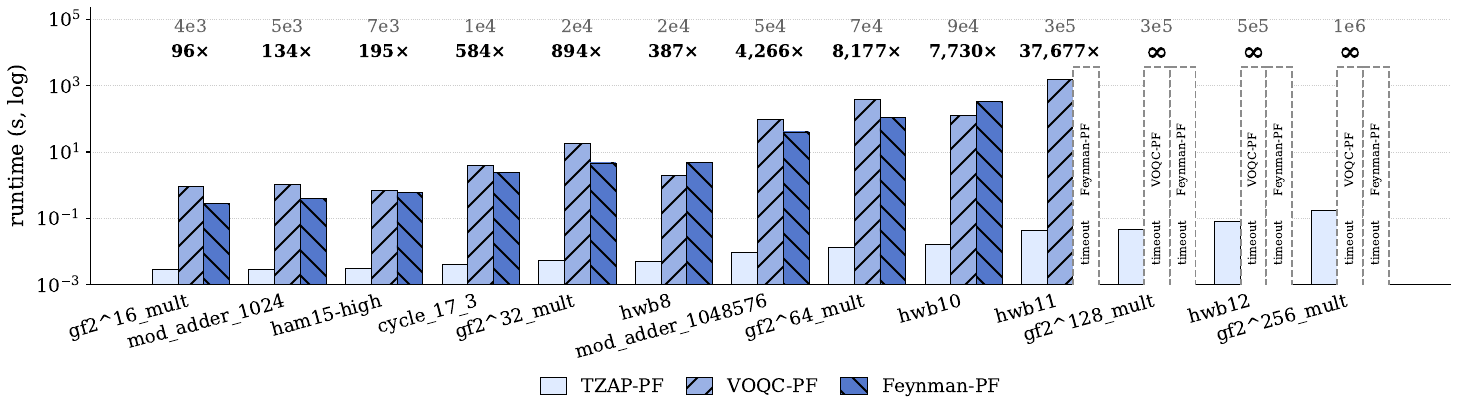}
    \caption{Runtime (\textbf{log} scale). On top of each benchmark, we show the number of gates (gray) and the speedup of \ours over the fastest other tool (black).}
    \label{fig:feynman-pf-runtime}
  \end{subfigure}
  \caption{Phase-folding isolation on the Feynman benchmark suite: \ours
    (\ourspf) versus the phase-folding-only variants of VOQC (VOQC-PF) and
    Feynman (Feynman-PF). The circuits are sorted by total gate count. Missing
    bars indicate that the tool crashed or timed out ($> 1$\,hour) on that
    circuit.}
  \label{fig:feynman-pf}
\end{figure}

\section{Construction of \texttt{hwb}-like Scaling Circuits}
\label{app:hwb-like}

\Cref{sec:rq3} attributes the worst-case slowdown of the symbolic
variant to the \texttt{hwb} circuits, whose dense interleaving of
Hadamards and $\CX$ gates produces high-weight parity formulas.
\Cref{alg:hwb-like} distills that structure into a parameterized family
in which the largest parity grows \emph{linearly} with circuit size,
isolating the effect for scaling experiments.

The construction designates $\q{q_0}$ as an \emph{accumulator} that
never receives a Hadamard. Each stage refreshes a round-robin
\emph{worker} qubit with a Hadamard---minting a fresh variable
$u_i$---and immediately XORs it into the accumulator with a $\CX$.
Because every $u_i$ is distinct, nothing ever cancels: after stage $i$
the accumulator's parity is $u_0 \oplus \cdots \oplus u_i$ (together
with its initial variable), so its weight grows by exactly one per
stage.  The trailing $\CX\ \q{q_0}\ \q{w}$ XORs the
accumulated parity into the worker, leaving it with a high-weight
parity (the current fresh variable cancels)---mirroring how
\texttt{hwb} spreads large parities across qubits---and this worker
parity is erased at the worker's next Hadamard.

\begin{algorithm}[H]
\caption{Generator for \texttt{hwb}-like circuits whose largest parity
  scales linearly with gate count.}
\label{alg:hwb-like}
\begin{algorithmic}[1]
\Require number of qubits $n \geq 2$; number of stages $s$
\Ensure circuit $\circuit$ with $5s$ gates whose largest parity has
  weight $s + 1$
\State $\circuit \gets \epsilon$
\For{$i = 0, \ldots, s - 1$}
  \State $\q{w} \gets \q{q_{1 + (i \bmod (n-1))}}$
    \Comment{round-robin worker qubit}
  \State append $\Had\ \q{w}$
    \Comment{fresh variable: $\hat{\sigma}(\q{w}) \gets u_i$}
  \State append $\CX\ \q{w}\ \q{q_0}$
    \Comment{$\hat{\sigma}(\q{q_0}) \gets \hat{\sigma}(\q{q_0}) \oplus u_i$}
  \State append $\Rz{\pi/4}\ \q{q_0}$
    \Comment{rotation on a weight-$(i+2)$ parity}
  \State append $\CX\ \q{q_0}\ \q{w}$
    \Comment{XOR the accumulated parity into the worker}
  \State append $\Rz{-\pi/4}\ \q{w}$
\EndFor
\State \Return $\circuit$
\end{algorithmic}
\end{algorithm}

On this family the symbolic analysis pays a cost proportional to the
accumulated parity's weight on every stage---the $\CX$ transfer merges
the variable sets, and each rotation keys a lookup on one---so the
total time grows quadratically with circuit size, while the randomized
abstraction folds every parity into a fixed-width bitstring and remains
linear. The number of qubits $n$ only sets how many gates separate two
visits to the same worker; the parity growth is independent of it.

\end{document}